\documentclass[reqno,11pt]{amsart}

\usepackage{amssymb,graphics,graphicx,wrapfig}

\usepackage{epsfig, psfrag, wrapfig}

\newcommand{\floor}[1]{\left\lfloor #1 \right\rfloor}
\def\eps{\varepsilon}

\def\N{\bbN}

\def\be{\begin{equation}}
\def\ee{\end{equation}}
\def\ba{\begin{align}}
\def\bm{\begin{multline}}
\def\bfig{\begin{figure}[htb]}
\def\efig{\end{figure}}

\numberwithin{equation}{section}
\newtheorem{theorem}{Theorem}[section]
\newtheorem{proposition}[theorem]{Proposition}
\newtheorem{lemma}[theorem]{Lemma}
\newtheorem{corollary}[theorem]{Corollary}

\newtheorem{definition}{Definition}

\DeclareMathSymbol{\leqslant}{\mathalpha}{AMSa}{"36}
\DeclareMathSymbol{\geqslant}{\mathalpha}{AMSa}{"3E}
\DeclareMathSymbol{\doteqdot}{\mathalpha}{AMSa}{"2B}
\DeclareMathSymbol{\circlearrowright}{\mathalpha}{AMSa}{"08}
\DeclareMathSymbol{\subsetneq}{\mathalpha}{AMSb}{"28}
\DeclareMathSymbol{\supsetneq}{\mathalpha}{AMSb}{"29}
\renewcommand{\leq}{\;\leqslant\;}
\renewcommand{\geq}{\;\geqslant\;}

\newcommand{\dd}{{\rm d}}
\newcommand{\e}[1]{\,{\rm e}^{#1}\,}
\newcommand{\ii}{{\rm i}}

\def\Im{{\operatorname{Im\,}}}

\newcommand{\upchi}{\raise 2pt \hbox{$\chi$}}

\def\writefig#1 #2 #3 {\rlap{\kern #1 truecm \raise #2 truecm
\hbox{#3}}}

\newcommand{\caO}{{\mathcal O}}

\newcommand{\caW}{{\mathcal W}}

\newcommand{\bbC}{{\mathbb C}}

\newcommand{\bbN}{{\mathbb N}}

\newcommand{\bbR}{{\mathbb R}}

\newcommand{\bsI}{{\boldsymbol I}}

\newcommand{\bsphi}{{\boldsymbol\phi}}
\newcommand{\bspsi}{{\boldsymbol\psi}}

\newcommand{\sx}{\boldsymbol\sigma_{\rm x}}
\newcommand{\sy}{\boldsymbol\sigma_{\rm y}}
\newcommand{\sz}{\boldsymbol\sigma_{\rm z}}
\newcommand{\bsone} {\boldsymbol 1}
\newcommand{\moy}{\#}
\newcommand{\commut}[2]{\left[ #1 , #2 \right]}
\newcommand{\moycom}[2]{\commut{#1}{#2}_{\moy}}
\newcommand{\moycomord}[3]{\commut{#1}{#2}_{\moy, #3}}
\newcommand{\antcom}[2]{\left[ #1 , #2 \right]_+}

\newcommand{\R}{\bbR}
\newcommand{\C}{\bbC}
\newcommand{\Or}{\caO}
\newcommand{\epsft}[1]{{\widehat{#1}^\eps}}
\newcommand {\norm} [2] [] {\ensuremath{ \left\Vert  #2  \right\Vert_{#1} } }
\newcommand{\facnorm}[2]{\norm[(#2)]{#1}}

\begin{document}


\title[Born-Oppenheimer transitions]{Superadiabatic transition histories\\\ in quantum molecular dynamics}

\author{Volker Betz, Benjamin D.\ Goddard and Stefan Teufel}
\address{Volker Betz \hfill\newline
\indent Department of Mathematics \hfill\newline
\indent University of Warwick \hfill\newline
\indent Coventry, CV4 7AL, England \hfill\newline
{\small\rm\indent http://www.maths.warwick.ac.uk/$\sim$betz/}
}
\email{v.m.betz@warwick.ac.uk}

\address{Benjamin Goddard \hfill\newline
\indent Department of Mathematics \hfill\newline
\indent University of Warwick \hfill\newline
\indent Coventry, CV4 7AL, England \hfill\newline
{\small\rm\indent http://www.warwick.ac.uk/staff/B.D.Goddard/}
}
\email{b.d.goddard@warwick.ac.uk}

\address{Stefan Teufel \hfill\newline
\indent Universit\"at T\"ubingen \hfill\newline
\indent Auf der Morgenstelle 10 \hfill\newline
\indent 72076 T\"ubingen, Germany \hfill\newline
{\small\rm\indent http://www.maphy.uni-tuebingen.de/members/stte}
}
\email{stefan.teufel@uni-tuebingen.de}

\maketitle

\begin{quote}
{\small
{\bf Abstract.}
We study the dynamics of a molecule's nuclear wave-function near an avoided crossing of two electronic energy levels, for one nuclear degree of freedom. We derive the general form of the Schr\"odinger equation in the $n$-th superadiabatic representation for all $n \in \bbN$,
and give some partial results about the asymptotics for large $n$. Using these results, we obtain closed formulas
for the time development of the component of the wave function in an initially unoccupied energy subspace, when a wave packet crosses the transition region. In the optimal superadiabatic representation, which we define, this component builds up monontonically. Finally, we give an explicit formula for the transition wave function away from the crossing, which is
in excellent agreement with high precision numerical calculations.
}  

\vspace{1mm}

\vspace{1mm}
\noindent
{\footnotesize {\it 2000 Math.\ Subj.\ Class.:} 81V55, 34E20}
\end{quote}

\section{Overview and main results}

We consider the dynamics of a nuclear wave packet as it travels
through an avoided crossing of electronic energy levels. This problem has a long history in mathematics, physics and
theoretical chemistry, and we feel that the reader should know both that history and some detail about the problem itself
in order to appreciate our contribution. However,
such a presentation would overload the introduction, and we shift it to a separate section.
Here, we concentrate on giving a quick overview of our results, with little or no explanations. We will aim to give all necessary explanations later.

We consider a time-dependent, two level Schr\"odinger equation of the form
\be \label{basic equation}
\ii \eps \partial_t \psi(x,t) = H \psi(x,t), \quad \text{with }\quad H = - \frac{\eps^2}{2} \partial_x^2 \bsI + V(x).
\ee
Above, $\psi \in L^2(\bbR,\bbC^2)$, $\bsI$ is the two-dimensional unit matrix, and the $2 \times 2$-matrix $V(x)$ is the ('diabatic') electronic potential energy. We will usually write $V$ in the form
\be \label{basic V}
V(x) = \rho(x) \left( \begin{array}{ll} \cos \theta(x) & \sin \theta(x) \\ \sin \theta(x) & -\cos \theta(x)\end{array}\right).
\ee
\eqref{basic V} contains the implicit assumption that $V$ is traceless. We assume that the system exhibits an avoided crossing of the adiabatic energy levels. Then a wave packet that is originally entirely in one
adiabtic subspace undergoes non-adiabatic transitions when it travels through the avoided crossing region. In the adiabatic representation, these transitions are of order $\eps$ globally, but they are  exponentially small in $\eps$ in the scattering region. It is known that there exist improved, superadiabatic, representations, such that the transitions are exponentially small for all times; but no viable way of explicitly  determining these was known. Our first main result is Theorem \ref{h_n lowest order}.
There we show how to compute, to leading order in $\eps$, the $n$-th superadiabatic representation of the equation \eqref{basic equation}. The symbols of the off-diagonal coupling elements are obtained via a 
a recursive set of differential equations, given in Proposition \ref{p coefficient recursion}.

The superadiabatic representations form an asymptotic series, i.e.\ they diverge for fixed $\eps$ when $n \to \infty$. Therefore it is natural to look for the optimal superadiabatic representation, where the relevant quantities are minimal as functions of $n$ for given $\eps$. To find it, and prove error bounds, one needs precise control on the large $n$ asymptotics of superadiabatic coupling elements. We do not know how to obtain this control at present. Section \ref{s4} contains some partial results, including the asymptotics of superadiabatic coupling elements for large momenta, cf.\ Theorem \ref{large p dominance}.

In Section \ref{sec trans hist} we apply the previous results in order to compute non-adiabatic transitions. In contrast to the previous sections, the results of this final section are non-rigorous. 
The main result of Section \ref{sec trans hist}, and in some respect of the whole paper, is an explicit formula for  the exponentially small wave packet that makes the non-adiabatic transition in terms of data that are local in space and time. Let $\epsft{\psi_{+,0}}(k)$ be the Fourier transform of the wave packet moving according to the Born-Oppenheimer approximation in the upper electronic surface at the time when its maximum reaches the transition point. If at that time one starts a wave packet with Fourier transform given by 
\be  \label{form1}
\epsft{\psi_{-}}(k) = {\rm sgn}(k) \chi_{\{k^2>4\delta\}}  \sin \left(\frac{\pi \gamma}{2}\right) \e{- \frac{q_{\rm c}}{\eps} | k-v(k)|}   
  \left(1 + \tfrac{k}{v(k)} \right)   \epsft{\psi_{+,0}}(v(k)) 
  \ee
in the lower electronic surface and again evolves it according to the Born-Oppenheimer approximation, then a leading order approximation to the true time-evolution, i.e.\ to the solution of Schr\"odinger's equation \eqref{basic equation}, is achieved shortly after the wave packets leave the transition region. Here 
\[
v(k) := {\rm sgn} (k) \sqrt{k^2-4\delta} 
\]
is the momentum before the transition and $2\delta$ the energy gap. 
Formula  (\ref{form1}) holds for the special case of constant eigenvalues $\pm\delta$, i.e.\ $\rho(x) = \delta$ in (\ref{basic V}).  It provides a practical way to correctly include non-adiabatic transitions at avoided crossings into the Born-Oppenheimer approximation. An analogous formula for Landau-Zener like avoided crossings will be the content of a future publication. 

While the derivation of (\ref{form1}) is not completely rigorous and many approximations enter, we show it to be in excellent agreement with high precision ab initio numerical calculations, for a wide range of parameters including moderately large $\eps$. Not only does it correctly describe the transition probability, but it yields the shape of the wave function after the transition to such a high accuracy that, e.g.\ for the situation plotted in Figure~\ref{fig:non-gauss}, the relative error for $\epsft{\psi_{-}}(k)$ from formula (\ref{form1}) and from a highly accurate numerical solution 
of the Schr\"odinger equation is of the order $10^{-5}$ uniformly in those $k$ on which 
$\epsft{\psi_{-}}(k)$ is essentially supported. Put differently, Figure~\ref{fig:non-gauss} shows both the true solution and the approximation by our formula (\ref{form1}).

The method to derive (\ref{form1}) is to investigate the full time evolution of the transition wave packet in an optimal superadiabatic basis. In the time-adiabatic simplification of the problem, it is known \cite{Be90, BT05-2} that there exists an optimal superadiabatic representation in which transitions are (uniformly in time) exponentially small, and have the universal shape of an error function. For Born-Oppenheimer transitions, it is not immediately clear how to 
identify optimality of superadiabatic representations. In Definition \ref{opt superad} we propose 
a novel, natural criterion for optimality. We show that, if our criterion is met, transition wave packets in the Fourier representation have the universal error-function shape as functions in time near their maximum. We also show that in a simple but important case our criterion can be met. 

The structure of our paper is as follows. In Section \ref{sect history}, we describe the basic ideas of the time-dependent Born-Oppenheimer approximation, and the main developments in the study of avoided crossings. In Section \ref{subs rep operators} we recall the construction of the adiabatic representation, which will help to understand the subsequent construction of superadiabatic representations in the remainder of Section \ref{sect super rep} and Section \ref{s3}.
Section \ref{s4} is devoted to asymptotic results. Section \ref{sec trans hist} contains the applications to nonadiabatic transitions and has been discussed in detail in the previous paragraph.

\section{Introduction and history} \label{sect history}

Time-dependent Born-Oppenheimer theory is the single most important tool for studying the quantum dynamics of molecules, just as the time-independent  Born-Oppenheimer approximation \cite{BO27} is for the theory of molecular bound states. The basic physical idea is that the electrons, being at least 2000 times lighter than the nuclei, move
much more rapidly and thus quickly adjust their position with respect to the nuclei. In particular, if they start in the $n$-th
bound state (for some fixed positions of nuclei), they should remain in the $n$-th bound state even though the nuclei are slowly moving;
of course this will then be the $n$-th bound state with respect to
the updated position of the nuclei. In turn, the nuclear quantum dynamical motion is determined by an effective potential
given by the energy level of the $n$-th bound state of the electrons, as a function of nuclear position. All of this is expected to be true
up to errors of order $\eps$, the ratio of masses between electrons and nuclei.

Although named after Born and Oppenheimer,
the time-dependent version of the theory was first proposed by Fritz London in \cite{Lon28}. The mathematical investigation started with
the work of Hagedorn \cite{Hag80}, which made the ideas of London precise to leading order in $\eps$, under the assumption of smooth inter-particle potentials. Since then there has been considerable progress. In particular, the approximation has been pushed to arbitrary order in $\eps$, and extended to cover the case of an isolated subset of the electronic spectrum as opposed to just one single band \cite{MS02, ST01, Teu03}.

Instead of trying to give a full account of all further contributions, we refer to the corresponding section of the excellent review article \cite{HJ06} by George Hagedorn and Alain Joye.

The subject of this article is the time development of transitions
between molecular energy levels, i.e.\ the dynamics of the part of the wave-function not obeying the Born-Oppenheimer approximation. These deviations from the adiabatic behaviour are  of immense interest in quantum chemistry. In their most extreme form they occur when two energy levels of a molecule cross at a given configuration of nuclei, in which case Born-Oppenheimer theory breaks
down completely and there are transitions of order one. This mechanism is now widely accepted as governing many important chemical reactions \cite{Yar01}, and is an active research topic both in the mathematical community \cite{Ha94, LT05} and in theoretical chemistry \cite{CGB05, BH06}. Here, we will have nothing to say about it and instead concentrate on the related topic of avoided crossings of energy levels. The Born-Oppenheimer approximation then holds to leading order (and, as it will turn out, even beyond all orders in the scattering regime), but the remaining deviations are still of great interest.
In fact, it is an avoided crossing situation that leads to the photo-dissociation of NaI, which is one of the paradigmatic chemical reactions in photochemistry \cite{RRZ89, Zew94}.

Let us consider a diatomic molecule, such as NaI.
After discarding centre of mass movement and rotational degrees of freedom, of the nuclear degrees of freedom only the internuclear distance $x$ remains. For notational simplicity we ignore the spin degrees of freedom, and  so the Hamiltonian is given by
\begin{equation} \label{diatom}
H = -\frac{\hbar^{2}}{2 M} \partial_{x}^{2}  -\frac{\hbar^{2}}{2 m_{\rm e}} \Delta_{y} + V_{\rm n}(x) + V_{\rm{en}}(x,y) + V_{\rm e}(y),
\end{equation}
where $M$ is the reduced mass of the nuclei, $m_{\rm e}$ is that of the electrons, $y = (y_{1}, \ldots y_{n})$ are the positions of the electrons, and $V_{\rm n}$, $V_{\rm{e}}$ and $V_{\rm{en}}$ are effective nuclear repulsion, electronic repulsion and nuclear-electronic attraction, respectively. We simplify (\ref{diatom}) further by using atomic units ($m_{\rm e} = \hbar = 1$), putting $M =\eps^{-2}$, and subsuming everything except the nuclear kinetic energy into the ``electronic hamiltonian'' $H_{\rm e}(x)$, acting as an operator in $L^{2}(\dd y)$ for each $x$. As a result, (\ref{diatom}) reads
\begin{equation} \label{diatom2}
H = -\frac{\eps^{2}}{2} \partial_{x}^{2} + H_{\rm e}(x),
\end{equation}
and on the time scale where a nontrivial nuclear motion can be observed the time-dependent Schr\"odinger equation is given by
\begin{equation} \label{time dep}
\ii \eps \partial_{t} \psi(x,y,t) = H \psi (x,y,t), \quad \psi(x,y,0) = \psi_{0}(x,y).
\end{equation}
Let us now assume that the $n$-th eigenvalue of $H_{\rm e}(x)$, say $E_{n}(x)$, is separated from the rest of the spectrum by a finite gap for all $x$, and non-degenerate. Let $y \mapsto \chi_{n}(x,y)$ be the corresponding eigenvector in $L^{2}(\dd y)$ for each $x$.
Then by Born-Oppenheimer theory we know that if we start the
Schr\"odinger evolution with
$\psi_{0}(x,y) = \phi_{0}(x) \chi_{n}(x,y)$, the solution
of (\ref{time dep}) at time $t$ will be given (up to errors of order $\eps$) by $\psi(x,y,t) = \phi(x,t) \chi_{n}(x,y)$, with $\phi(x,t)$ determined by
\begin{equation} \label{effdyn}
\ii \eps \partial_{t} \phi(x,t)  = -\frac{\eps^{2}}{2} \partial_{x}^{2} \phi(x,t)+ E_{n}(x) \phi(x,t) , \quad \phi(x,0) = \phi_{0}(x).
\end{equation}

In this work we are not interested in the part that stays within the the energy band $E_{n}$, but rather want to study the orthogonal complement (in $L^{2}(\dd y)$) of it, i.e.\ the part that makes the transition. To do so, we introduce a significant simplification of the model: we assume that the $n$-th electronic energy level interacts with one and only one other electronic energy level,
say the $(n-1)$-st. The Hamiltonian for these energy levels alone is then given by
\begin{equation} \label{2level}
H = -\frac{\eps^{2}}{2} \partial_{x}^{2} + V(x) \quad \mbox{with} \quad V(x) = \left( \begin{array}{ll} X(x) & Z(x) \\ Z^{\ast}(x) & - X(x) \end{array}\right),
\end{equation}
where we assume that $X(x)$ and $Z(x)$ are analytic in a strip containing the real axis, and that $\rho(x) = \sqrt{X^{2}(x)+Z^{2}(x)} > c > 0$. The latter corresponds to $E_{n}(x)$ being isolated from $E_{n-1}(x)$ in the full electronic Hamiltonian. $H$ is now an operator in $L^{2}(\dd x ,\bbC^{2})$, and  $\bbC^{2}$ corresponds to the subspace of $L^{2}(\dd y)$ spanned by the eigenvectors $\chi_{n}$ and
$\chi_{n+1}$. The approximation contained in (\ref{2level}) can be justified by arguing that if
the remaining energy levels are separated from $E_{n}$ and $E_{n+1}$
by a gap that is uniformly larger than the minimal distance between $E_{n}$ and $E_{n+1}$, then their interaction with the model system should be negligible. In general, we ought to be very careful with such claims, since the effects we are looking for are exponentially small. While it is conceivable that, by using space-adiabatic perturbation theory beyond all orders, one could rigorously justify that (\ref{diatom}) gives the same transitions as (\ref{2level})  does in the relevant bands with exponential accuracy, there certainly exists no proof. Thus we think of (\ref{2level}) as an uncontrolled approximation and take it as the starting point of our investigations.
Note that even among the two-level systems, (\ref{2level}) is not
the most general, since we forced the trace of $V$ to vanish. While including a nonzero trace may lead to interesting
effects, the case of vanishing trace is technically easier, and we defer the treatment of
the general case to future work.

It is instructive to simplify (\ref{2level}) even more by prescribing a ('classical') path $x(t)$ for the nuclear motion of the two-level system, instead of considering its quantum evolution.
The equation then reduces to
\begin{equation}
\label{time adiab}
\ii \eps \partial_{t} \phi(t)  = \left( \begin{array}{ll} \tilde{X}(t) & \tilde Z(t) \\ {\tilde Z}^{\ast} (t) & -\tilde{X}(t) \end{array}\right) \phi(t) =: V(t) \phi(t),
\end{equation}
which is to be read as an equation for the occupation probabilities
of the two relevant electronic energy bands given a certain motion
of the nuclei. The limit $\eps\to 0$ in equation (\ref{time adiab}) is known as the adiabatic limit. The adiabatic theorem states that any solution $\phi(t)$ starting at time $t_0$ in an eigenstate $\psi(t_0)$ of  the matrix $V(t_0)$ remains an eigenstate of $V(t)$ also at time $t$ up to terms of order $\eps$ as long as $V(t)$ has two distinct eigenvalues. One says that transitions between the adiabatic subspaces are of order $\eps$ in the presence of a spectral gap.
However, in 1932, Zener \cite{Zen32} investigated an explicitly solvable instance of (\ref{time adiab}), namely
$X(t) = t/2$ and $Z(t) = \delta/2$. He observed that by solving (\ref{time adiab}) with an initial condition parallel to the eigenvector corresponding to $+\rho(t)$ at $t=-\infty$, the solution at $t=+\infty$ would have a component of magnitude $\e{-\pi \delta^{2} / (4 \eps)}$ in the eigenspace corresponding to $-\rho(t)$. In other words, in the scattering regime, the transition amplitude is exponentially small, much smaller than the bound obtained in the standard adiabatic theorem. Shortly after, Landau \cite{Lan65} argued that the same exponentially small expression should describe the scattering regime also for general analytic $X(t)$ and $Z(t)$ such that, at the minimum $t_{0}$ of $\rho^{2}(t) = X^{2}(t) + Z^{2}(t)$, $V$ is to first order approximated by the one considered by Zener. This gave rise to the famous Landau-Zener formula, which in itself has attracted much research, including \cite{DOS78} where more general situations leading
to different prefactors where considered in generality for the first time, and \cite{JKP91} where a rigorous proof of the generalized Landau-Zener transitions was given for the first time. We refer to \cite{HJ06-2} for further information on the subject.

Quantitatively the exponentially small scattering amplitude cannot be explained easily by any method involving just the adiabatic subspaces, i.e.\ the instantaneous eigenspaces of $V(t)$. Indeed, when starting the evolution in one of the eigenspaces and monitoring the component
$\phi_{2}(t)$ parallel to the other one, one will typically observe a build-up of $|\phi_{2}(t)|$ up to order $\eps$, and later an eventual decay to the exponentially small final value. See \cite{LiBe91, BT06}
for a numerical illustration of the phenomenon. One way to understand
better what is actually going on is the complex WKB method: one solves  (\ref{time adiab}) not
on the real line, but on a curve in the complex plane that approaches the real line at $\pm \infty$ and passes through the zeroes of the complex continuation of $\rho(t)$. Then the solution is exponentially small all the way, and by this method Joye et.\ al.\ \cite{JKP91} prove the Landau-Zener formula. However, the method does not give any insight on what happens on the real line at finite times. This question was first
investigated by M.\ Berry. In an influential paper \cite{Be90} he expands the solution of (\ref{time adiab}) into a formal power series, and by truncating the resulting asymptotic series after $n$ terms, he obtains a time-dependent basis of $\bbC^{2}$, called $n$-th superadiabatic basis. For all $n$, these bases agree with the adiabatic basis in the scattering regime, i.e.\ 
when the eigenspaces are approximately constant. 
Choosing $n$ so that the remainder term in the asymptotic expansion is minimal, Berry shows that not only are transitions between the corresponding subspaces exponentially small, but they are also universal in the sense that to leading order they are described by an error function for a wide class of matrices $V(t)$. Later, Berry and Lim \cite{BeLi93} refined these results by showing that some non-generic versions of $V(t)$ would lead to different prefactors for the scattering amplitude. Both of these works were non-rigorous, and it was not until a decade later that Hagedorn and Joye \cite{HaJo04} succeeded to prove them in a special, non-generic case. Results in the same special case were obtained independently in \cite{BT05-1}, but with a different method: Instead of expanding the solution of (\ref{time adiab}), the equation itself
is transformed using adiabatic perturbation theory \cite{Teu03}, leading to superadiabatic representations, in analogy to superadiabatic bases. Each superadiabatic representation is given by a
a unitary matrix $U_{n}(t) \in \bbC^{2 \time 2}$ such that
the off-diagonal elements of the matrix valued operator
$U_n(t) (\ii \eps \partial_{t} - V(t)) U_n^{\ast}(t)$ are of order
$\eps^{n}$, but with a prefactor growing like a factorial.
The optimal superadiabatic basis is the one where the off-diagonal elements are minimal, and there they are shown to be
exponentially small Gaussian functions to leading order. Integrating the resulting effective equations using first order perturbation theory leads to Berry's results. It turned out that this method was sufficiently flexible to allow for a generalisation \cite{BT05-2} covering all cases of interest, in particular the generic one and most of those appearing in \cite{BeLi93}. So now (\ref{time adiab}) can be viewed as pretty much well understood.

The situation is quite different for the true Born-Oppenheimer evolution (\ref{2level}). While it is known that it is possible (even for more general models) to construct exponentially accurate solutions \cite{HJ01, MS02, So03}, all of these works merely give upper bounds on the error terms instead of computing them as functions of time to leading order, as was achieved in \cite{Be90} for (\ref{time adiab}). Thus they say nothing about universality of transition histories, quantitative values of Landau-Zener scattering amplitudes, or scattering wave functions. Concerning the last point, there has been recent progress by Hagedorn and Joye \cite{HaJo05}. They assume that the time evolution is started with a coherent semiclassical wavepacket, of sufficiently high momentum, located near 
$x=-\infty$, and contained in one energy band. Under various assumptions, they prove that 
the portion of the wave-function making the transition to the other energy band is exponentially small and approximately Gaussian; they even give explicit formulas for these wave packets in the scattering region, in terms of complex contour integrals.  Of note, these formulas show that the exponential
rate is larger than predicted by the corresponding Landau-Zener formula derived from (\ref{time adiab}), while the momentum of the transmitted wave-packet is larger than predicted by energy conservation. Intuitively this is due to the fact that fast portions of the wave packet are more likely to make the transition than slow ones, and the presence of those (as opposed to a perfectly sharp momentum assumed in the approximation leading to (\ref{time adiab})) increases both the exponential rate and the momentum of the transition wave. The method of their proof is again the complex WKB method. Since the results of 
\cite{HaJo05} are of relevance to our asymptotic results, we will discuss the relation to our findings 
towards the end of Section \ref{constant energy transitions}. 

\section{Superadiabatic representations} \label{sect super rep}

\subsection{Representation as operators} \label{subs rep operators}
Our starting points are \eqref{basic equation} and \eqref{basic V}.
Switching to a superadiabatic representation means that we need to find a pseudo-differential operator
$U_n$ in $L^2(\bbR, \bbC^2)$ such that $U_n H U_n^\ast$ is close to a diagonal (operator-valued) matrix in a suitable sense. As suggested by the notation, $U_n$ will be close to a unitary. The benefit of such a representation is that then, with $\bsphi_n = U_n \bsphi$,  the equation
\be \label{superadiab general}
\ii \eps \partial_t \bsphi_n = U_n H U_n^\ast \bsphi_n
\ee
will decouple, up to error terms that we control, into two scalar equations. Moreover, an explicit control of the off-diagonal coupling terms in $U_n H U_n^\ast$ will yield results on the superadiabatic transition histories simply through first order time-dependent perturbation theory. The most well-known instance of this procedure is the adiabatic transformation obtained by
\be \label{U0}
U_0(x) = \left( \begin{array}{cc} \cos (\theta(x)/2) & \sin (\theta(x)/2) \\ \sin (\theta(x)/2) & -\cos (\theta(x)/2) \end{array}\right).
\ee
Since $U_0$ diagonalizes $V$, we obtain
\be \label{superad operator}
U_0 H U_0^\ast = - \frac{\eps^2}{2} \partial_x^2 +  \left( \begin{array}{cc} \rho(x) + \eps^2 \frac{\theta'(x)^2}{8} & 
- \eps  \frac{  \theta'(x)}{2} \cdot (\eps \partial_x)  - \eps^2 \frac{\theta''(x)}{4} \\[2mm]
 \eps  \frac{  \theta'(x)}{2} \cdot (\eps \partial_x)  + \eps^2 \frac{\theta''(x)}{4}  & 
-\rho(x) + \eps^2 \frac{\theta'(x)^2}{8} \end{array}\right).
\ee
Since $U_0 H U_0^\ast$ acts on semiclassical wave functions that oscillate with frequency $1/\eps$, the operator
  $\eps \partial_x$ is actually of order one, and to leading order we get
\[
H_0 :=  - \frac{\eps^2}{2} \partial_x^2 +  \left( \begin{array}{cc} \rho(x)  &- \eps  \frac{  \theta'(x)}{2} \cdot (\eps \partial_x)  \\[2mm]   \eps  \frac{  \theta'(x)}{2} \cdot (\eps \partial_x) & -\rho(x) \end{array}\right),
\]
the standard expression for the adiabatic dynamics including the well-known derivative coupling. In particular, $H_0$ gives the correct leading order dynamics inside the energy bands. However, if we are interested in exponentially small transitions, the adiabatic representation is not sufficient, and we need to find different representations (i.e.\ unitary transformations of $L^2(\R,\C^2)$) such that
the powers of $\eps$ in the off-diagonal increase.

\subsection{Symbolic representation}
For this, we will be dealing with high order differential operators, and it is convenient to work in the symbolic representation. A full account of symbolic calculus can be found in \cite{DiSj99,Ma02}. Here we only outline the main formulae that we are going to use, and do not touch the topic of symbol classes. Working in the symbolic representation means that we first replace $x$ by $q \in \R$ and $\ii \eps \partial_x$ by an independent variable $p \in \R$ in the definition (\ref{basic equation}) of $H$. The factor $\eps$ takes into account the semiclassical scaling. We then obtain
\be
H(p,q) = \frac{p^2}{2} + V(q),
\ee
where $V$ is as in (\ref{basic V}). The aim is now to find matrices $U_n(\eps,p,q)$ such that
\be \label{H_n}
H_n(\eps,p,q) := U_n(\eps,p,q) \moy H(p,q) \moy U_n(\eps,p,q) = \rm{diag} + \caO(\eps^{n+1}).
\ee
Here $\moy$ denotes the\\[1mm]
\noindent {\bf Moyal Product: } For two symbols $A(p,q)$ and $B(p,q)$, their Moyal product is defined through $A \moy B = \sum_j \eps^j (A \moy B)_j$ with
\begin{equation} \label{genmoyal}
(A(p,q) \moy B(p,q))_j = (2 \ii)^{-j} \sum_{\alpha + \beta = j}
\frac{(-1)^{\alpha}}{\alpha! \beta!} \left( \partial_q^\alpha \partial_p^\beta A \right) \left( \partial_p^\alpha \partial_q^\beta B \right).
\end{equation}
The Moyal product is the natural product for semiclassical symbols and accounts for the non-commutativity of the operators they represent.
It can be extended to symbols depending on $\eps$, like $U_n$, by representing $U_n$ as a formal power series $U_n(\eps,p,q) = \sum_{k} \eps^k U_{n,k}(p,q)$ and collecting powers of $\eps$, but we will only ever need the formula (\ref{genmoyal}) in explicit calculations.

Once $H_n(\eps,p,q)$ is constructed,
it can be recast as a pseudo-differential operator using the\\[1mm]
\noindent {\bf Weyl-Quantisation: } the operator
$\caW_\eps(H_n)$ corresponding to $H_n(\eps,p,q)$ acts on a test function $\phi$ by
\\
\be \label{weyl quant}
(\caW_{\eps}(H_n) \phi) (x) = \frac{1}{2 \pi \eps} \int_{\R^2} \dd \xi \, \dd y \, H_n \left(\eps, {\textstyle \frac{1}{2}}(x+y),\xi\right)
\e{\frac{\ii}{\eps}\xi(x-y)}\phi(y).
\ee

With the connection between operators and symbols in place, we now need to  actually construct the superadiabatic unitaries $U_n$. For this we use the method of\\[1mm]
\noindent {\bf Superadiabatic Projections: } that is, we seek symbols\\
\be \label{pi^n}
\pi^{(n)}(\eps,p,q) = \sum_{j=0}^n \eps^j \pi_j(p,q) \quad \text{(with $\pi_j(p,q) \in \C^{2 \times 2}$),}
\ee
such that
\begin{eqnarray}
\left( \pi^{(n)}(p,q) \right)^{\moy 2} - \pi^{(n)}(p,q) & = & \eps^{n+1} G_{n+1} + \Or(\eps^{n+2}) \label{pi1} \\
\moycom{H(p,q)}{\pi^{(n)}(p,q)} & = & \eps^{n+1} F_{n+1} + \Or(\eps^{n+2}). \label{pi2}
\end{eqnarray}
Above, $\moycom{A}{B} = A \moy B - B \moy A$ denotes the Moyal commutator.

General theory \cite{Teu03} guarantees the existence of $\pi^{(n)}$ for each $n$, and the same general theory states

\begin{lemma} \label{general fact}
There exists a semiclassical symbol $U_n(\eps,p,q)$ such that
\begin{eqnarray}
U_0(q) \moy U_n(\eps,p,q) = U_0(q) U_n(\eps,p,q) + \Or(\eps) & = & \bsone + \Or(\eps), \label{appr0} \\
U_n^{\ast}(\eps,p,q) \moy U_n(\eps,p,q) &=& \bsone + \Or(\eps^{n+1}), \label{appr1}\\
U_n(\eps,p,q) \moy \pi^{(n)}(\eps,p,q) \moy U_n^\ast(\eps,p,q) &=& \pi_{\rm r}, \label{appr2}
\end{eqnarray}
where $\pi_{\rm r}$ is the projection onto the first component of $\bbC^2$.
\end{lemma}

While the existence of the $U_n$ is known, they are usually tricky to calculate.
Fortunately, there is no need to do so for getting the leading order of the superadiabatic Hamiltonian.
This is the content of the following result.

\begin{proposition} \label{offdiag Fn}
Define $H_n(\eps,p,q)$ as in (\ref{H_n}), and assume the $U_n$ given in that formula fulfils
(\ref{appr0})--(\ref{appr2}). Let $F_n$ be given through (\ref{pi2}). Denote by
\[
c_{n}^+(p,q) = (U_0(q) F_{n}(p,q) U_0(q))_{1,2}
\]
the upper off-diagonal element of $U_0 F_{n} U_0$, and let $c_n^-(p,q)$ be the lower off-diagonal element. Then
\[
H_n(\eps,p,q) =
\begin{pmatrix} p^2/2 + \rho(q) & \eps^{n+1} c_{n+1}^+(p,q) \\
- \eps^{n+1} c_{n+1}^-(p,q)  & p^2/2 - \rho(q) \end{pmatrix} + \begin{pmatrix} \Or(\eps^2) & \Or(\eps^{n+2}) \\ \Or(\eps^{n+2}) & \Or(\eps^2) \end{pmatrix}.
\]
\end{proposition}

\begin{proof}
The diagonal terms are immediate from first order adiabatic perturbation theory. For the off-diagonal terms, we multiply (\ref{appr2}) with $U_n$ from the right and use (\ref{appr1}) in order to find
\be \label{projcomm1}
\pi_{\rm r} U_n =  U_n \moy \pi^{(n)} \moy U_n^* \moy U_n  = U_n \moy \pi^{(n)} + \eps^{n+1} U_n \moy \pi^{(n)} \moy R
\ee
for some symbol $R$, and similarly for $U_n \pi_{\rm r}$.
Iterating the above reasoning, we find
\be \label{projcomm2}
U_n \moy \pi^{(n)} = \pi_{\rm r} U_n - \eps^{n+1} U_n \moy \pi^{(n)} \moy R = \pi_{\rm r} U_n
\moy (1 - \eps^{n+1} R) + \Or(\eps^{2n+2}).
\ee
We now use (\ref{projcomm1}) in order to get
\begin{eqnarray*}
\lefteqn{(1-\pi_{\rm r}) \moy U_n \moy H \moy U_n^\ast \moy \pi_{\rm r}} \\
&=& (1-\pi_{\rm r}) \moy U_n \moy H \moy \pi^{(n)} \moy U_n^{\ast} \moy (1+\Or(\eps^{n+1})) = \\
&=& (1-\pi_{\rm r}) \moy U_n \moy \left(\pi^{(n)} \moy H + \eps^{n+1} F_{n+1} + \Or(\eps^{n+2})\right) \moy U_n^{\ast} \moy (1+\Or(\eps^{n+1})) = \\
&=& (1-\pi_{\rm r}) \moy ( \pi_{\rm r} \moy U_n + \eps^{n+1} U_n \moy \pi^{(n)} \moy R) \moy H  \moy U_n^{\ast} \moy (1+\Or(\eps^{n+1})) +\\
&& + (1-\pi_{\rm r}) \moy U_n \moy \left(\eps^{n+1} F_{n+1} + \Or(\eps^{n+2})\right) \moy U_n^{\ast} \moy (1+\Or(\eps^{n+1})).
\end{eqnarray*}
The next to last line above is $O(\eps^{2n + 2})$ by (\ref{projcomm2}) and the fact $(1-\pi_{\rm r}) \pi_{\rm r} = 0$,
and multiplying with $\pi_{\rm r}$ from the right we find
\[
(1-\pi_{\rm r}) \moy U_n \moy H \moy U_n^\ast \moy \pi_{\rm r}  =
\eps^{n+1} (1-\pi_{\rm r}) U_0 (F_{n+1} + \Or(\eps)) U_0 \pi_{\rm r}.
\]
This is the result for the upper right off-diagonal element. For the lower left one, we interchange the roles of $\pi_{\rm r}$ and $1 - \pi_{\rm r}$, and obtain an additional minus sign from (\ref{pi2}).
\end{proof}
Proposition \ref{offdiag Fn} means that we can focus our attention entirely on superadiabatic projections,
and there will be no need to calculate superadiabatic unitaries.

\section{Superadiabatic projections} \label{s3}

Here we present the recursive scheme for calculating superadiabatic projections, getting more and more explicit as the
section progresses.

\subsection{Matrix recursion for superadiabatic projections}
It is easy to check that for
\[
\pi_0(p,q) = \pi_0(q) = \frac{1}{2} \big( 1 + V(q)/\rho(q) \big),
\]
we have $\pi_0 V = \rho \pi_0$.  Hence $\pi_0$ is the adiabatic (zeroth superadiabatic) projection, i.e. the projection on the eigenspace of $V$ corresponding to $+ \rho$. Projecting on the upper adiabatic subspace is just a matter of choice here, the same construction works when starting with the projection onto the other subspace.

In order to obtain the higher superadiabatic representations, we note that
\[
F_{n+1} = \moycomord{\frac{p^2}{2} + V(q)}{\pi^{(n)}(p,q)}{n+1} = \moycomord{p^2/2}{\pi_n}{1} +
\sum_{k=1}^n \moycomord{V}{\pi_{n+1-k}}{k}.
\]
Here, $\moycomord{A}{B}{k} = (A \moy B)_k - (B \moy A)_k$ is the coeffcient of $\moycom{A}{B}$ corresponding to $\eps^k$. Using (\ref{genmoyal}), we find
\be \label{F_n+1}
F_{n+1} = \frac{1}{\ii} p \partial_q \pi_n + \sum_{k=1}^n \frac{1}{(2 \ii)^k k!} \left( (-1)^k (\partial_q^k V) \, (\partial_p^{k}\pi_{n+1-k}) - (\partial_p^{k}\pi_{n+1-k}) \, (\partial_q^k V) \right).
\ee
Similarly, (\ref{appr2}) gives
\be \label{G_n+1}
G_{n+1} = \sum_{k=1}^{n} \pi_k \pi_{n+1-k} + \sum_{k=0}^n \left( \pi_k \moy \pi_{n-k} \right)_1 + \sum_{k=0}^{n-1} \left( \pi_k \moy \pi_{n-k-1} \right)_2 + \ldots.
\ee
Finally, we can calculate $\pi_{n+1}$ through
\be \label{pi_n recursion}
\pi_{n+1} = G_{n+1} - \pi_0 G_n - G_n \pi_0 + \frac{1}{2 \rho} \commut{F_{n+1}}{\pi_0}.
\ee
The proof of \eqref{pi_n recursion} is given in \cite{BT05-1}, Proposition 1, for the special case $\rho = 1/2$ and
$F_{n+1} = - \ii \pi_n'(q)$. Since that proof applies word by word to the general case,
we do not repeat it here. Similarly, the important relations
\be \label{F_n is offdiag}
\pi_0 F_{n} \pi_0 = (1-\pi_0) F_n (1-\pi_0) = 0
\ee
and
\be \label{G_n is diag}
(1-\pi_0) G_n \pi_0 = \pi_0 G_n (1-\pi_0) = 0
\ee
follow as in \cite{BT05-1}.

\subsection{Transformed Pauli matrices}
We need a more explicit recursive scheme in order to feasibly calculate the quantities $F_n$, $G_n$ and $\pi_n$ given in (\ref{pi^n}) -- (\ref{appr2}). We now introduce such a scheme,  following \cite{BT05-1}, but will use a different notation than \cite{BT05-1} for reasons to be discussed below.
Let
\be \label{raw pauli}
\sigma_{\rm x} = \left( \begin{array}{ll} 0 & 1 \\ 1 & 0 \end{array} \right), \quad \sigma_{\rm y} = \left( \begin{array}{ll} 0 & -\ii \\ \ii & 0 \end{array} \right), \quad \sigma_{\rm z} = \left( \begin{array}{ll} 1 & 0 \\ 0 & -1 \end{array} \right)
\ee
be the Pauli matrices. We will need their representations in the adiabatic basis, given by
\be \label{pauli}
\sx(q) = U_0(q) \sigma_{\rm x} U_0(q), \quad \sy(q) = U_0(q) \sigma_{\rm y} U_0(q), \quad \sz(q) = U_0(q) \sigma_{\rm z} U_0(q).
\ee
Above, note that $U^\ast_0 = U_0$. The usual algebraic relations of the un-transformed Pauli matrices give
\be \label{pauli algebraic}
\begin{split}
\sx \sy = - \sy \sx = \ii \sz, \qquad & \sx \sz = - \sz \sx = - \ii \sy, \\
\sy \sz = - \sz \sy = \ii \sx, \qquad & \sx^2 = \sy^2 = \sz^2 = \bsone,
\end{split}
\ee
where $\bsone$ is the unit matrix.
Moreover, the special relations
\be \label{pauli special}
\begin{array}{lll}
\sx' = \theta' \sz, \qquad & \sy' = 0, \qquad & \sz' = - \theta' \sx,\\[2mm]
\commut{\sx}{\pi_0} = -\ii \sy, \qquad & \commut{\sy}{\pi_0} = \ii \sx, \qquad & \commut{\sz}{\pi_0} = 0.
\end{array}
\ee
can be easily checked. Here and henceforth, primes denote derivatives with respect to $q$.\\[2mm]
\noindent {\bf Remark:} We will use the Pauli matrices as a basis to represent the $\pi_n$ in.
The basic idea is the same as in \cite{BT05-1}, but there the basis matrices $X,Y,Z$ were chosen in an ad-hoc manner.
It happens that the matrices $X,Y,Z$ from \cite{BT05-1} are linked to $\sx, \sy, \sz$ here,
but unfortunately not in the most convenient way: We have $X = \ii \sy$, $Y = - \sz$ and $Z = - \sx$.
This will lead to a serious clash of notation between the present paper and \cite{BT05-1}, but we feel that the formulation of the problem in the widely used Pauli matrices justifies this.

\subsection{Superadiabatic projections through Pauli matrices}
Let us first note that $\sz(q) = V(q) / \rho(q)$. This implies that the adiabatic projection is given by
\be \label{pi0}
\pi_0(p,q) = \pi_0(q) = \frac{1}{2} \left( 1 + \sz(q) \right).
\ee
Indeed, $\pi_0^{\moy 2} = \pi_0^2 = \pi_0$ due to $\sz^2 = U_0^2 = \bsone$, and
\be \label{adiabatic offdiagonal}
\moycom{H(p,q)}{\pi_0(q)} = \moycom{{\textstyle \frac{p^2}{2}} \, \bsone}{\pi_0(q)} = \eps \frac{p}{\ii} \partial_q \pi_0(q) =
\frac{\eps \ii p}{2} \theta'(q) \sx(q).
\ee
The final equality above follows from the fact $\sz'(q) = - \theta'(q) \sx(q)$.
In particular, we conclude that
\[
F_1(p,q) = \frac{\ii}{2} p \theta'(q) \sx(q),
\]
and Proposition \ref{offdiag Fn} then gives
\[
H_0(\eps,p,q) = \left( \begin{array}{ll} p^2/2 + \rho(q) & \eps \ii p \theta'(q)/2 \\
- \eps \ii p \theta'(q)/2 & p^2/2 - \rho(q) \end{array} \right) + \Or(\eps^2),
\]
which is just the adiabatic representation \eqref{superad operator} in symbolic language. \\

We now define the coefficients $x_n$ through $w_n$ by
\be \label{pi_n decomposition}
\pi_n(p,q) = x_n(p,q) \sx(q) + \ii y_n(p,q) \sy(q) + z_n(p,q) \sz(q) + w_n(p,q) \bsone,
\ee
and emphasize that these coefficients have swapped names with those given in
\cite{BT05-1}. The prefactor $\ii$ in front of $y_n$ will make for some more elegant formulas later, and slightly reduce the clash of notation with \cite{BT05-1}. For our derivation of the recursions for the coefficients $x_n$ to $w_n$, we first need to treat the
derivatives of $V$ appearing in (\ref{F_n+1}):

\begin{lemma} \label{V derivatives}
We have
\[
\partial_q^n V(q) = a_n(q) \sz(q) + b_n(q) \sx(q),
\]
where $a_n(q)$ and $b_n(q)$ are given by the recursions
\be \label{anbn}
\begin{array}{rclrcl}
a_0(q) &=& \rho(q), \quad & b_0(q) &=& 0 \\[2mm]
a_{n+1}(q) &=& a_n'(q) + \theta'(q) b_n(q), \qquad & b_{n+1}(q) &=& b_n'(q) - \theta'(q) a_n(q).
\end{array}
\ee
\end{lemma}
\begin{proof}
Use the fact $V = \rho \sz$ together with (\ref{pauli special}).
\end{proof}

The first step towards the coefficient recursion is

\begin{proposition} \label{f_n expression}
Define $x_n,y_n,z_n$ and $w_n$ through (\ref{pi_n decomposition}). Then
\[
\begin{split}
& F_{n+1} = \left( \frac{p}{\ii} (x_n' - \theta' z_n) - 2 \left( \frac{\partial_p}{(2\ii) 1!} b_1 w_n - \frac{\partial^2_p}{(2\ii)^2 2!} a_2 y_{n-1} + \frac{\partial^3_p}{(2\ii)^3 3!} b_3 w_{n-2} - \ldots \right) \right) \sx 
\\
 & \, + \left( \frac{p}{\ii} y_n' - 2 \left( \frac{\partial^2_p}{(2\ii)^2 2!} (b_2 z_{n-1} - a_2 x_{n-1}) +
\frac{\partial^4_p}{(2\ii)^4 4!} (b_4 z_{n-3} - a_4 x_{n-3}) + \ldots \right) \right) \ii \sy  \\
 & \, + \left( \frac{p}{\ii}(z_n' + \theta' x_n) + 2 \left( - \frac{\partial_p}{(2\ii) 1!} a_1 w_{n} -
\frac{\partial^2_p}{(2\ii)^2 2!} b_2 y_{n-1} - \frac{\partial^3_p}{(2\ii)^3 3!} a_3 w_{n-2} - \ldots
\right) \right) \sz \\
 & \, + \left( \frac{p}{i} w_n' - 2 \left( \frac{\partial_p}{(2 \ii) 1!} (a_1 z_n + b_1 x_n) +
\frac{\partial_p^3}{(2 \ii)^3 3!}(a_3 z_{n-2} + b_3 x_{n-2}) + \ldots \right) \right) \bsone.
\end{split}
\]
\end{proposition}

\begin{proof}
This is just a calculation. The former parts of each bracket stem from the first term of
(\ref{F_n+1}), which computes to
\[
\partial_q \pi_n = (x_n' - \theta' z_n) \sx + y_n' \sy + (z_n' + \theta' x_n) \sz + w_n' \bsone.
\]
The latter terms of each bracket come from the various terms of second part of (\ref{F_n+1});
for even $k$ those are given by
\[
\begin{split}
& \frac{\partial_p^k}{(2\ii)^k k!} \commut{\partial_q^k V}{\pi_{n+1-k}} = \\
& = \frac{\partial_p^k}{(2\ii)^k k!} \commut{a_k \sz + b_k \sx}{x_{n+1-k} \sx + \ii y_{n+1-k} \sy + z_{n+1-k} \sz + w_{n+1-k} \bsone} = \\
& = -\frac{\partial_p^k}{(2\ii)^k k!} 2 \ii \left( \ii a_k y_{n+1-k} \sx + (b_k z_{n+1-k} - a_k x_{n+1-k}) \sy - \ii b_k y_{n+1-k} \sz \right),
\end{split}
\]
while for odd $k$ they are
\[
\begin{split}
& - \frac{\partial^k_p}{(2\ii)^k k!} \antcom{\partial^k_q V}{\pi_{n+1-k}} =\\
& = -\frac{\partial_p^k}{(2\ii)^k k!} \antcom{a_k \sz + b_k \sx}{x_{n+1-k} \sx + \ii y_{n+1-k} \sy + z_{n+1-k} \sz + w_{n+1-k} \bsone} = \\
& = - \frac{\partial_p^k}{(2 \ii)^k k!} 2 \left( (a_k z_{n+1-k} + b_k x_{n+1-k}) \bsone +
a_k w_{n+1-k} \sz + b_k w_{n+1-k} \sx \right).
\end{split}
\]
Collecting coefficients gives the result.
\end{proof}

It is remarkable that from this knowledge of $F_{n+1}$ alone, we can derive a set of recursive differential equations that, together with zero boundary conditions at infinity, determine the coefficients $x_n$ to $w_n$:

\begin{proposition} \label{coefficient recursion}
The coefficients $x_n$ to $w_n$ defined in (\ref{pi_n decomposition}) are determined by the following recursive algebraic-differential equations: we have
\be \label{rec start}
x_1 = z_1 = w_1 = 0, \qquad y_1 = -\ii p \frac{\theta'(q)}{4 \rho(q)}.
\ee
Moreover,
\be \label{rec zeros}
y_n = 0 \quad \text{when $n$ is even,} \qquad x_n = z_n = w_n = 0 \quad \text{when $n$ is odd.}
\ee
For $n$ odd we have
\be \label{x_n eq1}
x_{n+1} = -\frac{1}{2 \rho} \left( \frac{p}{\ii} y_{n}' - 2 \sum_{j=1}^n \frac{\partial_p^j}{(2 \ii)^j j!}
(b_j z_{n+1-j} - a_j x_{n+1-j}) \right),
\ee
while for $n$ even we have
\begin{eqnarray}
y_{n+1} & = & -\frac{1}{2 \rho} \left( \frac{p}{\ii} (x_n' - \theta' z_n) - 2 \sum_{j=1}^n
\frac{\partial_p^j}{(2 \ii)^j j!} (- a_j y_{n+1-j} + b_j w_{n+1-j} ) \right), \label{y_n eq1}\\
0 & = & \frac{p}{\ii} (z_n' + \theta' x_n) - 2 \sum_{j=1}^n \frac{\partial_p^j}{(2 \ii)^j j!}
(b_j y_{n+1-j} + a_j w_{n+1-j}) \label{z_n eq1} \\
0 & = & \frac{p}{\ii} w_n' - 2 \sum_{j=1}^n \frac{\partial_p^j}{(2\ii)^j j!} (a_j z_{n+1-j} + b_j x_{n+1-j})
\label{w_n eq1}
\end{eqnarray}
\end{proposition}

\begin{proof}
For (\ref{rec start}), note that
\[
\pi_1 = \frac{1}{2 \rho} \commut{F_1}{\pi_0} = \frac{\ii p \theta'}{4 \rho} \commut{\sx}{\pi_0} =
\frac{p \theta'}{4 \rho} \sy = \left( -\ii p \frac{\theta'}{4 \rho} \right) \ii \sy
\]
We now use (\ref{pi_n recursion}) and Proposition \ref{f_n expression}
in order to prove the recursive formulae: from (\ref{G_n is diag}), we can deduce that
$ \commut{G_{n+1}}{\pi_n} = 0$, and thus (\ref{pauli special}) implies that $G_{n+1}$ is proportional
to $\sz$ and $\bsone$. Consequently, the parts of $\pi_{n+1}$ proportional to $\sx$ and $\sy$ arise from the last term of
(\ref{pi_n recursion}) alone, and comparison with (\ref{pauli special}) and Proposition \ref{f_n expression}
shows
\be \label{raw x_n}
x_{n+1} = -\frac{1}{2 \rho} \left( \frac{p}{\ii} y_n' - 2 \left( \frac{\partial_p^2}{(2 \ii)^2 2!}(b_2 z_{n-1} - a_2 x_{n-1})
+ \frac{\partial^4_p}{(2 \ii)^4 4!} (b_4 z_{n-3} - a_4 x_{n-3}) + \ldots \right) \right),
\ee
and
\be \label{raw y_n}
y_{n+1} = -\frac{1}{2 \rho} \left( \frac{p}{\ii} (x_n' - \theta' z_n) - 2 \left( \frac{\partial_p}{(2 \ii) 1!} b_1 w_n
- \frac{\partial^2_p}{(2 \ii)^2 2!} a_2 y_{n-1} + \frac{\partial^3_p}{(2 \ii)^3 3!} b_3 w_{n-2} - \ldots
\right) \right).
\ee
Now, from (\ref{F_n is offdiag}) we can deduce that $F_n$ is proportional to $\sx$ and $\sy$ only, and so
Proposition \ref{f_n expression} immediately gives
\begin{eqnarray}
0 & = &  \frac{p}{\ii}(z_n' + \theta' x_n) - 2 \left( \frac{\partial_p}{(2\ii) 1!} a_1 w_{n} +
\frac{\partial^2_p}{(2\ii)^2 2!} b_2 y_{n-1} + \frac{\partial^3_p}{(2\ii)^3 3!} a_3 w_{n-2} + \ldots
\right), \label{raw z_n},\\
0 & = & \frac{p}{i} w_n' - 2 \left( \frac{\partial_p}{(2 \ii) 1!} (a_1 z_n + b_1 x_n) +
\frac{\partial_p^3}{(2 \ii)^3 3!}(a_3 z_{n-2} + b_3 x_{n-2}) + \ldots \right).
\label{raw w_n}
\end{eqnarray}
Now (\ref{rec zeros}) follows inductively, and after that (\ref{x_n eq1}) -- (\ref{w_n eq1}) are immediate from (\ref{raw x_n}) -- (\ref{raw w_n}), noting that all we did is to add some terms that are zero, in order to get a more closed expression.
\end{proof}

We are now in the position to cast Proposition \ref{offdiag Fn} into a more specific form, yielding our first main result.

\begin{theorem} \label{h_n lowest order}
Define $H_n(\eps,p,q)$ as in Proposition \ref{offdiag Fn}, $x_n(p,q)$ and $y_n(p,q)$ as in
Proposition \ref{coefficient recursion}, and put
\be \label{kappa formula}
\kappa_{n+1}^{\pm} = -2 \rho (y_{n+1} \pm x_{n+1}).
\ee
Then
\[
H_n(\eps,p,q) = \frac{p^2}{2} \bsone + \begin{pmatrix}
\rho & \eps^{n+1} \kappa_{n+1}^{+}  \\ \eps^{n+1} \kappa_{n+1}^{-} & -\rho
\end{pmatrix} + \begin{pmatrix} \Or(\eps^2) & \Or(\eps^{n+2}) \\ \Or(\eps^{n+2}) & \Or(\eps^2) \end{pmatrix}.
\]
\end{theorem}
\begin{proof}
From Proposition \ref{offdiag Fn}, we see that we need to calculate $U_0 F_{n+1} U_0$. By definition, $U_0 \sx U_0 = \sigma_x$, and similarly for $\sy$. Since $F_{n+1}$ is proportional to $\sx$ and $\sy$, by comparing Propositions \ref{f_n expression} and \ref{coefficient recursion} we obtain
\[
U_0(q) F_{n+1}(p,q) U_0(q) = 2 \rho(q) y_{n+1}(p,q) \sigma_x + 2 \ii \rho(q) x_{n+1}(p,q) \sigma_y.
\]
Comparison with (\ref{raw pauli}) and Proposition \ref{offdiag Fn} gives the result.
\end{proof}

In order to make use of Theorem \ref{h_n lowest order}, we need to control the coefficients $x_n$ and $z_n$.
Since they are poynomials (of order $n$) in $p$, it makes sense to consider coefficients. We put
\be \label{xn order decomposition}
x_n(p,q) = \sum_{m=0}^{n} p^{n-m} x_n^m(q),
\ee
with similar expressions for the other coefficients. Here, the index $m$ in $x_n^m$ is an upper index rather than a power, and the choice of using $x_n^0$ for the highest power will turn out to be the most convenient one  below. Differentiating now gives
\[
\partial_p^j x_{n+1-j} = \sum_{m=0}^{n+1-2j} \frac{(n+1-j-m)!}{(n+1-2j-m)!} p^{n+1-2j-m} x_{n+1-j}^m,
\]
and thus the latter terms of expressions like (\ref{x_n eq1}) are of the form
\begin{align}
\sum_{j=1}^n \frac{\partial^j_p}{(2\ii)^j j!} a_j x_{n+1-j} & = \sum_{j=1}^{\floor{(n+1)/2}} \, \,
\sum_{k=0}^{n+1-2j} \frac{(n+1-j-k)! a_j}{(2\ii)^j j! (n+1-2j-k)!} p^{n+1-2j-k} x_{n+1-j}^k = \notag \\
& = \sum_{m=0}^{n+1} p^{n+1-m} \sum_{j=1}^{\floor{m/2}} \frac{a_j}{(2\ii)^j} \frac{(n+1-m+j)!}{j! (n+1-m)!}
x_{n+1-j}^{m-2j}. \label{xn derivatives}
\end{align}

\begin{proposition} \label{p coefficient recursion}
The coefficients $x_n^m$ to $w_n^m$ defined in (\ref{xn order decomposition}) are determined by the following recursive algebraic-differential equations: we have
\be \label{rec start 2}
x_1^m = z_1^m = w_1^m = 0, \; m=0,1, \qquad y_1^0 = -\ii \frac{\theta'(q)}{4 \rho(q)}, \; y_1^1=0.
\ee
Moreover,
\be \label{x_n eq2}
x_{n+1}^m = -\frac{1}{2 \rho} \left( \frac{1}{\ii} (y_n^m)' - 2 \sum_{j=1}^{\floor{m/2}} \frac{1}{(2\ii)^j}
\tbinom{n+1-m+j}{j} (b_j z_{n+1-j}^{m-2j} - a_j x_{n+1-j}^{m-2j}) \right)
\ee
for $n$ odd, while for $n$ even we have
\begin{eqnarray}
y_{n+1}^m &=& -\frac{1}{2 \rho} \left( \frac{1}{\ii} \big( (x_n^m)' - \theta' z_n^m \big) -
2 \sum_{j=1}^{\floor{m/2}} \frac{1}{(2\ii)^j}
\tbinom{n+1-m+j}{j} ( -a_j y_{n+1-j}^{m-2j} + b_j w_{n+1-j}^{m-2j}) \right) \nonumber ,\\
&&  \label{y_n eq2}\\
0 & = & \frac{1}{\ii} \big( (z_n^m)' + \theta' x_n^m \big) - 2 \sum_{j=1}^{\floor{m/2}} \frac{1}{(2\ii)^j}
\tbinom{n+1-m+j}{j} (b_j y_{n+1-j}^{m-2j} + a_j w_{n+1-j}^{m-2j}), \label{z_n eq2}\\
0 & = & \frac{1}{\ii} (w_n^m)' - 2 \sum_{j=1}^{\floor{m/2}} \frac{1}{(2\ii)^j}
\tbinom{n+1-m+j}{j} (a_j z_{n+1-j}^{m-2j} + b_j x_{n+1-j}^{m-2j}). \label{w_n eq2}
\end{eqnarray}
\end{proposition}
\begin{proof}
This simply uses (\ref{xn derivatives}) in (\ref{x_n eq1})--(\ref{w_n eq1}).
\end{proof}

\noindent {\bf Remark:} From the above equations, it is obvious that $x_n^m = y_n^m = z_n^m = w_n^m = 0$ for odd $m$. But more is true. By induction and using (\ref{rec zeros}), we find that
\be \label{m zeros}
\begin{split}
x_n^m = y_n^m = z_n^m = 0 & \quad \text{if} \quad m \neq 4 k \\\
w_n^m = 0 & \quad \text{if} \quad m \neq 4 k + 2,
\end{split}
\ee
for  $k \in \N_0$.

\section{Asymptotics} \label{s4}

It appears to be a very hard problem to determine the asymptotic behaviour of the quantities $x_n^m$ etc. as
$n \to \infty$, let alone prove it. So our results on this subject are rather incomplete, except in the case $m=0$,
i.e.\ for the term of the highest order (in $p$); in that case
the asymptotics of $x_n^0(q)$, $y_n^0(q)$ and $z_n^0(q)$ is known rigorously.

\subsection{Highest order in $p$}

For $m=0$, the sums on the right hand side of (\ref{x_n eq2})--(\ref{w_n eq2}) are empty, and we retain
\[
x_{n+1}^0 = \frac{\ii}{2\rho} (y_n^0)', \quad y_{n+1}^0 = \frac{\ii}{2\rho}( (x_n^0)' - \theta' z_n^0),
\quad 0 = (z_n^0)' + \theta' y_n^0.
\]
After changing to the natural scale
\be \label{nat}
\tau(q) = 2 \int_0^q \rho(r) \, \dd r,
\ee
these are (apart from a change of notation)
just the recursions appearing in the time-adiabatic case, which have been solved in \cite{BT05-1, BT05-2}.

The relevant results in that case read as follows: we introduce the natural scale (\ref{nat})
and define $\tilde f(\tau(q)) = f(q)$ for a given function $f$. Furthermore, we assume that
\be \label{nat scale}
\frac{\dd}{\dd \tau} \tilde \theta(\tau) = \frac{\ii \gamma }{\tau - \ii \tau_{\rm c}} - \frac{\ii \gamma }{\tau + \ii \tau_{\rm c}} + \tilde \theta_{\rm r}'(\tau),
\ee
for some $\gamma \in \R$ and $\tau_{\rm c} > 0$,
where $\tilde \theta_{\rm r}(\tau)$ has no singularities in $\{ z \in \bbC: |\Im(z)| \leq \tau_{\rm c} \}$, and only singularities of order smaller than one at $\pm \ii \tau_{\rm c}$.
As has first been observed in \cite{BeLi93}, \eqref{nat scale} is the form of $\tilde\theta$ 
for a large class of models, including the generic ones. For further discussion we refer to  \cite{BeLi93, BT05-2}.

From (\ref{rec start 2}) we now conclude
\[
y_1(q) = - \ii \frac{\theta'(q)}{4 \rho(q)} = - \ii \frac{\tilde \theta'(\tau(q))}{2} = - \frac{\ii}{2} \left(
\frac{\ii \gamma }{\tau(q) - \ii \tau_{\rm c}} - \frac{\ii \gamma }{\tau(q) + \ii \tau_{\rm c}} + \tilde \theta_{\rm r}'(\tau(q)) \right).
\]
Let us write $\kappa_n^{0,\pm}$ for the coefficient of $\kappa_n^{\pm}$ belonging to $p^n$. Then by \eqref{kappa formula},
we have
\[
\kappa_n^{0, \pm}(q) = \left\{ \begin{array}{ll} \mp 2 \rho(q)  	x_n^0(q) & \text{if $n$ is even,}\\
										- 2 \rho(q) 	y_n^0(q) & \text{if  $n$ is odd.}
											\end{array}
											\right.
\]
Then by the results from \cite{BT05-1, BT05-2}, there exists $\beta>0$ such that
\begin{eqnarray}
\kappa_n^{0,\pm}(q) &=& - \alpha(n) \rho(q) \ii^n (\pm 1)^{n+1}
\partial_{\tau}^n \tilde \theta(\tau) + (n-1)! \Or(n^{-\beta}))  \label{h_n^0}
\\
&=& \alpha(n) \rho(q) \ii^n (\pm \ii)^{n} (n-1)! \left( \frac{\ii \gamma }{(\tau(q) - \ii \tau_{\rm c})^n} - \frac{\ii \gamma }{(\tau(q) + \ii \tau_{\rm c})^n} \right) + (n-1)! \Or(n^{-\beta}) \nonumber\\
y_n^0(q) & = & (n-1)! \Or(n^{-\beta}). \nonumber
\end{eqnarray}
Above $\alpha(n) = \frac{\sin(\pi\gamma/2)}{\pi\gamma/2} (1 + \Or(1/n))$ is the universal prefactor for
the time-adiabatic transitions.

\subsection{Estimates for lower orders in $p$}

Unfortunately, the exact knowledge of the above asymptotics does not help us directly.
The reason is that the terms $p^n x_n^0(q)$ and $p^n y_n^0(q)$ appear not to constitute the leading order
contribution to $x_n(p,q)$ and $y_n(p,q)$ as $n \to \infty$. We expect this to be true in general, but it
can be verified in the Landau-Zener case
\cite{BG08}, where we have $\sup_q |x_n^m(q)| \sim c^m \Gamma(n+m/4) / \Gamma(n)$ for finite $m$ and large $n$. We
do not know the behaviour of $x_n^m$ when $m$ is of the order of $n$; however, since numerical calculations clearly
show that the latter case is where $x_n^m$ is maximal for fixed $n$, this is the regime that needs to be understood in order
to do exponential asymptotics.

In this work, all we can do is to prove a rough a-priori bound on the coefficients. While it is presumably not sharp, we will
see that it identifies regimes of momenta for the incoming wave function such that the highest order in $p$
determines the behaviour of the transition. We recall the norms introduced in \cite{BT05-1}, which we will use here.
For $\tau_{\rm c} > 0$ and $I \subset \R$, we define
\begin{equation} \label{norm}
\facnorm{f}{I,\alpha, \tau_{\rm c}} := \sup_{t \in I} \sup_{k \geq 0} \left| \partial^k f(t)
\right| \frac{\tau_{\rm c}^{\alpha + k}}{\Gamma(\alpha + k)} \leq \infty
\end{equation}
for a function $f \in C^\infty$ on the real line. We also define
\[
F_{\alpha,\tau_{\rm c}}(I) = \left \{ f \in C^{\infty}(I): \facnorm{f}{I,\alpha,\tau_{\rm c}} < \infty \right\}.
\]
When $\tau_{\rm c}$ and $I$ are fixed, we will simply write  $\facnorm{\cdot}{\alpha}$ and $F_\alpha$.
In \cite{BT05-1} we prove
\begin{gather}
\sup_{q \in I} \left| \partial^k f(q) \right| \leq \frac{\Gamma(\alpha+k)}{\tau_{\rm c}^{\alpha + k}}
\facnorm{f}{I,\alpha,\tau_{\rm c}} \quad \forall k\geq 0, \label{normprop 1} \\
\facnorm{f'}{I,\alpha+1,\tau_{\rm c}} \leq \facnorm{f}{I,\alpha,\tau_{\rm c}}, \label{normprop 2} \\
\facnorm{\int_s^t f(r) \, \dd r}{I,\alpha-1,\tau_{\rm c}} \leq \max \left\{ \frac{(\alpha-1) |t-s|}{\tau_{\rm c}}, 1 \right\} \facnorm{f}{\alpha}, \label{normprop 3}\\
\facnorm{fg}{I,\alpha+\beta,\tau_{\rm c}} \leq B(\alpha,\beta) \facnorm{f}{I,\alpha,\tau_{\rm c}}\facnorm{g}{I,\beta,\tau_{\rm c}},
\label{normprop 4}
\end{gather}
where $B(\alpha,\beta) = \Gamma(\alpha)\Gamma(\beta)/\Gamma(\alpha+\beta)$ is the Beta function.

For the following a-priori estimate on the coefficients, we only treat the special case of constant
eigenvalues, and without further loss of generality take $\rho = 1/2$. This is the case for which we will also derive the
explicit transitions and transition histories later; the important Landau-Zener case will be treated elsewhere \cite{BG08}.
By (\ref{nat scale}), we have $\theta' \in F_{(1)}$.

\begin{proposition} \label{higher m estimate}
Assume $\rho = 1/2$, and $\theta' \in F_{1}(I)$.
\begin{itemize}
\item[a)] $x_n^m \in F_{n}(I)$ for any interval $I$, and the same holds for
the other coefficients.
\item[b)] Assume further that $|I| \leq \min \{ \tau_{\rm c} / \facnorm{\theta'}{1},  \tau_{\rm c} / \facnorm{\theta'}{1}^2 \}$. Then for every $\alpha > 1/2$ there exists $C_\alpha > 0$ such that for all $m,n$ we have
\[
\facnorm{x_n^m}{I,n,\tau_{\rm c}}, \facnorm{y_n^m}{I,n,\tau_{\rm c}} \leq C_\alpha \frac{\Gamma(n+\alpha m)}{\Gamma(n)},
\]
and
\[
\facnorm{z_n^m}{I,n,\tau_{\rm c}}, \facnorm{w_n^m}{I,n,\tau_{\rm c}} \leq \frac{2 C_\alpha}{\facnorm{\theta'}{1}} \frac{\Gamma(n+\alpha m)}{\Gamma(n)}
\]
\end{itemize}
\end{proposition}

\begin{proof}
We will proceed inductively and use Proposition \ref{p coefficient recursion}. To this end, let us first note that by
Lemma \ref{V derivatives} and (\ref{normprop 2}), (\ref{normprop 4}), an easy induction shows that $a_n$ and $b_n$ are in
$F_{n}$, with $\facnorm{a_n}{n}, \facnorm{b_n}{n} \leq c_0 \ln (n+1)$ for some $c_0 > 0$. Thus, by Proposition
\ref{p coefficient recursion} and (\ref{normprop 2}), (\ref{normprop 4}), we have $x_n \in F_{n}$, and the same for all
other coefficients. It remains to give bounds on the actual size of the norms. For $m=0$, the bounds claimed in b)
were given in
\cite{BT05-1}, cf.\ Theorem 2 there. Now let us assume that the claim b) holds up to $m-1$ and up to $n$. Let us write
\[
S_n^m := \sum_{j=1}^{\floor{m/2}} \frac{1}{(2\ii)^j} \binom{n+1-m+j}{j} (b_j y_{n+1-j}^{m-2j} - a_j w_{n+1-j}^{m-2j}).
\]
Then
\[
\begin{split}
& \facnorm{S_n^m}{n+1} \leq  \sum_{j=1}^{\floor{m/2}} \frac{1}{2^j} \frac{\Gamma(n+2-m+j)}
{\Gamma(n+2-m) \Gamma(j+1) } \frac{\Gamma(n+1-j) \Gamma(j)}{\Gamma(n+1)}
\times \\
& \times (\facnorm{b_j}{j} \facnorm{y_{n+1-j}^{m-2j}}{n+1-j} +
\facnorm{a_j}{j} \facnorm{w_{n+1-j}^{m-2j}}{n+1-j}) \leq C_\alpha c_0 \left(1 + \frac{2}{\facnorm{\theta'}{1}}\right)
\times \\
& \times
\sum_{j=1}^{\floor{m/2}} \frac{1}{2^j} \frac{\ln j}{j} \frac{\Gamma(n+2-m+j) \Gamma(n+1-j) \Gamma(n+1+\alpha m-(1+2 \alpha)j)}
{\Gamma(n+2-m) \Gamma(n+1) \Gamma(n+1-j)}.
\end{split}
\]
Clearly, the fraction of Gamma functions with $j=1$ is the largest of all, and thus the factor $1/2^j$ allows to estimate the sum through twice its first term, giving
\[
\facnorm{S_n^m}{n+1} \leq C_\alpha \tilde C \frac{(n+2-m) \Gamma(n+\alpha m - 2 \alpha)}{\Gamma(n+1)},
\]
where $\tilde C$ does not depend on $n$ or $m$.
The same holds for all the other sums appearing in Proposition \ref{p coefficient recursion}. Using the recurrence relation there, we have
\[
\begin{split}
\facnorm{y_{n+1}^m}{n+1} & \leq \facnorm{x_n^m}{n} + \frac{\facnorm{\theta'}{1}}{n} \facnorm{z_n^m}{n} +
C_\alpha \tilde C \frac{(n+2-m) \Gamma(n+\alpha m - 2 \alpha)}{\Gamma(n+1)} \leq \\
& \leq C_\alpha \left( \frac{\Gamma(n+\alpha m)}{\Gamma(n)} + 2 \frac{\Gamma(n + \alpha m)}{\Gamma(n+1)} + \tilde C \frac{(n+2-m) \Gamma(n+\alpha m - 2 \alpha)}{\Gamma(n+1)} \right) = \\
& = C_\alpha \frac{\Gamma(n+1+\alpha m)}{\Gamma(n+1)} \left( \frac{n+2}{n+1+\alpha m} + \tilde C  \frac{(n+2-m) \Gamma(n+\alpha m - 2 \alpha)}{\Gamma(n+1 + \alpha m)}\right).
\end{split}
\]
The last term in the bracket above is $\Or((n+ \alpha m)^{- 2 \alpha})$, and as $\alpha > 1/2$, it vanishes faster than $1/n$.
Thus the bracket becomes smaller than one for large enough $n$. By choosing $C_\alpha$ so large that the induction hypothesis
holds up to this $n$, we have shown the induction step for $y_n^m$. The argument for $x_n^m$ is similar and simpler. As for $z_n^m$, we have
\[
\facnorm{{z_n^m}'}{n+1} \leq \frac{1}{n} \facnorm{\theta'}{1} \facnorm{x_n^m}{n} + C_\alpha \tilde C \frac{(n+2-m) \Gamma(n+\alpha m - 2 \alpha)}{\Gamma(n+1)},
\]
and using (\ref{normprop 3}), we see that
\[
\begin{split}
\facnorm{z_n^m}{n} & \leq C_\alpha \frac{|I|}{q_c} \left( \facnorm{\theta'}{1} \facnorm{x_n^m}{n} +  \tilde C \frac{(n+2-m) \Gamma(n+\alpha m - 2 \alpha)}{\Gamma(n)} \right)\\
& \leq C_\alpha \frac{\Gamma(n + \alpha m)}{\Gamma(n)} \left(1 +  \tilde C \frac{(n+2-m) \Gamma(n+\alpha m - 2 \alpha)}{\Gamma(n+ \alpha m)} \right).
\end{split}
\]
The last bracket will  be bounded by 2 for large enough $n$, and so the same reasoning as above shows the induction step
for $z_n^m$. The proof for $w_n^m$ is similar and simpler.
\end{proof}

By piecing together intervals as given in Proposition \ref{higher m estimate} b), and using (\ref{normprop 1}), we obtain
\begin{corollary} \label{large m growth}
For any compact interval $I$ and any $\alpha > 1/2$, there exists a constant $C$ such that
\[
\sup_{q \in I} |x_n^m(q)| \leq C \frac{\Gamma(n + \alpha m)}{\tau_{\rm c}^n},
\]
and the same for $z_n^m$.
\end{corollary}

\subsection{Coupling function for high momenta}
The final result of this section is about the asymptotic shape of the coupling functions $\kappa^{0,\pm}_n$ given in
Theorem \ref{h_n lowest order}, for high momenta. From (\ref{xn order decomposition}) it is clear that choosing $p$ large enough will suffice to counter the growth of the coefficients as given in Corollary \ref{large m growth}. To obtain precise
statements, let us first note that by choosing
\be \label{optimal n}
n = \frac{\tau_{\rm c}}{\eps p},
\ee
we obtain
\be \label{opt trunc}
\eps^n p^n x_n^0(q) = 2 \ii \sqrt{2 \eps}{\pi \tau_{\rm c}} \sin(\pi\gamma/2) \e{-\frac{\tau_{\rm c}}{\eps p}}
\e{- \frac{q^2}{2 \eps p \tau_{\rm c}}} \e{\frac{\ii}{\eps p} q} (1 + \Or((\eps p)^{1/2 - \delta})).
\ee
This follows directly from the results in \cite{BT05-1}. Recall also that under the assumptions $\theta' \in F_{1}$ and
$\rho = 1/2$, we have $p^n x_n^0 \sim p^n \frac{(n-1)!}{\tau_{\rm c}^n}$ at its maximum. We will now specify the regime of $p$ for which
this is also the leading order behaviour.

\begin{proposition} \label{large p dominance}
Assume $\rho = 1/2$ and $\theta \in F_{1}$, and $I \subset \R$ compact.
Assume further that $p = \eps^{- \beta}$, with $1/3 < \beta < 1$. Then there exists $\delta > 0$ such that
for all $n \leq \frac{\tau_{\rm c}}{p \eps}$, we have
\[
\eps^n \sum_{m=1}^n p^{n-m} x_n^m(q) = \eps^n p^n \frac{(n-1)!}{\tau_{\rm c}^n} \Or( \eps^\delta).
\]
\end{proposition}
\begin{proof}
By Corollary \ref{large m growth} and Stirlings formula, we have
\[
\left| \frac{\tau_{\rm c}^n}{p^n (n-1)!} \sum_{m=1}^n p^{n-m} x_n^m(q)  \right| \leq
\sum_{m=1}^n \frac{\Gamma(n+ \alpha m)}{\Gamma(n) p^m} \leq c \sum_{m=1}^n \frac{(n+ \alpha m)^{n+\alpha m}
\e{- \alpha m}}{n^n p^m}
\]
for any $\alpha > 1/2$. Clearly, this is largest for the maximal value $n = \frac{\tau_{\rm c}}{p \eps}$,
so it suffices to treat this case.
Inserting $\eps = p^{-1/\beta}$ into $n = \frac{\tau_{\rm c}}{p \eps}$ gives
\[
p = \tau_{\rm c}^{-\frac{\beta}{1 - \beta}} n^{\frac{\beta}{1 - \beta}}.
\]
Thus
\[
\begin{split}
\frac{(n+ \alpha m)^{n+\alpha m} }{n^n p^m} & \leq \frac{(n+ \alpha m)^{n+\alpha m} }{n^{n + \frac{\beta}{1 - \beta} m}}
\tau_{\rm c}^{\frac{\beta}{1 - \beta} m}\\
&= \exp \left( (n + \alpha m) \ln (n + \alpha m) - (n+ \frac{\beta}{1 - \beta}m) \ln (n) +
m \frac{\beta}{1-\beta} \ln \tau_{\rm c} \right).
\end{split}
\]
Now when $\beta > 1/3$, we can pick $\alpha > 1/2$ such that $\beta/(1-\beta) > \alpha$, and the exponent becomes
negative for large enough $n$, and all $m < n$. The factor $\e{- \alpha m}$ in the sum above then guarantees
summability up to $m=n$ without losing more than a constant, and the proof is finished.
\end{proof}

Together with (\ref{opt trunc}), the previous result immediately gives
\begin{corollary}
Define $\kappa^{0,\pm}_n$ as in Theorem \ref{h_n lowest order}.
We make the same assumptions as in Proposition \ref{large p dominance}, and put $n = \frac{\tau_{\rm c}}{p \eps}$. Then
\be \label{asympt coupling}
\kappa^{0,\pm}_n(p,q) =  \mp 2 \ii \sqrt{2 \eps}{\pi \tau_{\rm c}} \sin(\pi\gamma/2) \e{-\frac{\tau_{\rm c}}{\eps p}}
\e{- \frac{q^2}{2 \eps p \tau_{\rm c}}} \e{\frac{\ii}{\eps p} q} (1 + \Or(\eps^{\delta})),
\ee
for some $\delta > 0$.
\end{corollary}

This result is not completely satisfactory for several reasons. Firstly, in is only valid in the special case of constant
eigenvalues $\rho$. However, an extension to general forms of the potential energy is merely a matter of, albeit laborious,
routine. For the Landau-Zener model, this is addressed in a work in progress \cite{BG08}. Secondly, our result only holds
for momenta that scale with $\eps$, while ideally we would like to have asymptotics for fixed momentum. To get these appears
to be a formidable problem, as knowledge about the true asymptotic behaviour of the $x_n^m$ for higher $m$ is needed. We
have put many efforts into that problem, with little tangible results. Thirdly, and most importantly from the applied
point of view, it is not the  asymptotic shape of the coupling functions we are interested in, but rather the shape
of the time-dependent transmitted wave function. To obtain these, we will either need to solve the Wigner equation,
or translate back into the language of operators using the Weyl quantisation. In both cases, it is inconvenient that
the expression (\ref{asympt coupling}) contains terms of the form $\e{-q^2/(\eps p)}$, since these are in none of
the usual symbol classes \cite{DiSj99,Ma02}. There is no theory, and worse, no calculus for dealing with such symbols.

The way out of this dilemma is to {\em not} apply optimal truncation until the very end: we do a Weyl qunatisation
for finite $n$, solve the corresponding PDE, and decide in the end how large we want $n$ to be. This is the content of the final section.

\section{Transitions} \label{sec trans hist}

We now discuss transitions in the superadiabatic representations given in Theorem \ref{h_n lowest order}.
Since we do not have asymptotic information about the coefficients $x_n^m(q)$ for large $n$ and $m$, we cannot
treat the most desirable case of asymptotics in $\eps$ for fixed momentum $p$. Therefore, except in Section \ref{general transhist}, we will work in the spirit of the previous subsection and treat only large $p$. The approximation then consists in retaining only the terms $x_n^0$ in the
more explicit formulas we provide. It will turn out that for those values of $\eps$ where numerics
of the exact Schr\"odinger evolution are feasible,
our momentum does not have to be particularly large. 
These are also the cases where the result is not too small to be physically relevant. 
Concretely, we are talking about $\eps$ between $1/10$ and $1/50$, and $p$ between $2$ and $5$. 
We refrain from proving error estimates in this section, as they
would be difficult to obtain and then could only hold for large $p$; rigorous treatment of such a restricted case would add something to our mathematical understanding, but not enough for us to consider it worth the effort. The value to practical applications of the formal calculations below is, on the other hand, potentially large.

\subsection{General transition histories} \label{general transhist}
Our starting point is the Schr\"odinger equation in the $n$-th superadiabatic representation. Assuming that the potential $V(x)$ approaches a constant matrix at spatial infinity quickly enough
guarantees that all of the superadiabatic subspaces approach the adiabatic subspaces for large $x$. More precisely we assume for the moment that the limits $\lim_{q\to\pm\infty}\rho(q)$ exist, that
\[
\lim_{|q|\to\infty} \kappa_{n+1}^\pm(q,p) = 0
\]
and that the limits are approached sufficiently fast.
We now study transitions between the superadiabatic subspaces for solutions that are asymptotically, for $t\to-\infty$, in the upper adiabatic subspace.
More precisely, consider the Schr\"odinger equation
\be \label{plain schroedinger n}
\ii \eps \partial_t \bspsi(x,t) = \caW_{\eps}(H_n(\eps,q,p)) \bspsi(x,t)\,,
\ee
where
  $\caW_\eps (H_n)$ is the Weyl quantization of the symbol 
 \[
 H_n = \left(\begin{array}{cc} 
\frac{p^2}{2}+ \rho + \Or(\eps^2) & 0\\[2mm]
0 &\frac{p^2}{2} - \rho + \Or(\eps^2)
 \end{array}\right) +
  \left(\begin{array}{cc} 
 0 & \eps^{n+1}\kappa_{n+1}^+ \\[2mm]
  \eps^{n+1}\kappa_{n+1}^-   &0
 \end{array}\right)+ \Or(\eps^{n+2})\,.
 \]
  Here we split $H_n(q,p)$, and in the same way also $\caW_\eps (H_n)$, into a diagonal part   of order one and an off diagonal ``coupling'' part   of order $\eps^{n+1}$. 
   Writing
\[
\bspsi(t,x) = \left( \begin{array}{cc} \psi_{+}(t,x) \\ \psi_{-}(t,x) \end{array} \right)\,,
\]
we can now use first order time-dependent perturbation theory in order to determine $\psi_-(t,x)$  up to errors of order $\eps^{n+1}$:
The perturbative solution of (\ref{plain schroedinger n}) with initial datum
$\psi_+(0,x) = \psi_{+,0}(x)$ and boundary condition $\lim_{t\to-\infty} \|\psi_-(t)\| = 0$ is
\[
\psi_+(t,x) =( \e{-\frac{\ii}{\eps} H_+t}\psi_{+,0})(x)  +\Or(\eps)
\]
and 
\be \label{psi_- sol}
\psi_-(x,t) = - \ii \eps^{n} \int_{-\infty}^t \left( \e{-\frac{\ii}{\eps}(t-s) H_-} \big( \caW_\eps \kappa_{n+1}^{-} \big)
\e{-\frac{\ii}{\eps} s H_+} \psi_{+,0} \right) (x) \, \dd s + \Or(\eps^{n+1})\,.
\ee
Here 
\[
\kappa_{n+1}^{-}(p,q) = - 2\rho(q) (- x_{n+1}(p,q) + y_{n+1}(p,q))
\]
is as given in Theorem \ref{h_n lowest order}, and 
\[
\quad H_\pm \psi(x) = \left(-  {\textstyle \frac{\eps^2}{2}}\partial_x^2 \pm \rho(x)\right) \psi(x)\,.
\]



In view of formula (\ref{psi_- sol}), we need the Weyl-quantisation of $\kappa_{n+1}^{-}(p,q)$.
Given that $\kappa_{n+1}^{-}$ is a polynomial in $p$, we will be interested in a general formula for the Weyl quantisation
of a symbol of the form $p^m g(q)$. Moreover, we will later need the Fourier representation of the operator
$\caW_\eps(\kappa_{n+1}^{-})$.\\[1mm]
To this end we define the {\bf scaled Fourier transform}
\begin{equation} \label{eps ft}
\epsft{f}(k) = \frac{1}{\sqrt{\eps}}\, \hat f\left(\frac{k}{\eps}\right) = \frac{1}{\sqrt{2 \pi \eps}} \int \e{-\frac{\ii}{\eps}kx} f(x) \, \dd x.
\end{equation}

\begin{lemma} \label{weyl p^m g}
Let $\kappa(p,q) = p^m g(q)$ be a semiclassical symbol for some $m \in \N$. Then
\be \label{weyl quant p^m g}
\epsft{(\caW_\eps \kappa) \psi}(k) = \frac{1}{\sqrt{2 \pi \eps}} \int_\R \epsft{g}(k-\eta)
\left( \frac{\eta + k}{2} \right)^m \epsft \psi(\eta) \dd \eta.
\ee
\end{lemma}

\begin{proof}
We use (\ref{weyl quant}) with
\[
\psi(y) = \frac{1}{\sqrt{2 \pi \eps}} \int \e{\ii \eta y / \eps} \epsft{\psi}(\eta) \, \dd \eta
\]
in order to get
\begin{eqnarray*}
(\caW_\eps \kappa) \psi (x) &=& \frac{1}{(2 \pi \eps)^{3/2}} \int_{\R^3} \dd y \, \dd \xi \, \dd \eta \,
g \left({\textstyle \frac{1}{2}}(x+y) \right) \xi^m \e{\frac{\ii}{\eps} ( \xi(x-y) + \eta y)} \epsft{\psi}(\eta) = \\
&=& \frac{2}{(2 \pi \eps)^{3/2}} \int_{\R^3} \dd y \, \dd\xi \, \dd\eta \, g(y) \xi^m \e{\frac{\ii}{\eps}( \xi x + (2y - x)(\eta-\xi))} \epsft\psi(\eta) = \\
& = & \frac{2}{2 \pi \eps} \int_{\R^2} \dd \xi \, \dd \eta \, \epsft{g} \big( 2(\xi-\eta) \big) \xi^m
\e{\frac{\ii}{\eps} (2 \xi - \eta)x} \epsft \psi(\eta) = \\
& = & \frac{1}{2 \pi \eps} \int_{\R^2} \dd \xi \, \dd \eta \, \epsft{g} \left( \xi - 2 \eta \right) \left( \frac{\xi}{2} \right) ^m
\e{\frac{\ii}{\eps} (\xi - \eta)x} \epsft \psi(\eta).
\end{eqnarray*}
In the second line we changed variables $\tilde y = (x+y)/2$, and in the fourth $\tilde \xi = 2\xi$. We now apply the scaled Fourier transform to both sides of the above equations and use the formula
\be \label{epsft delta}
\frac{1}{\sqrt{2 \pi \eps}} \int \dd x \, \e{\frac{\ii}{\eps}(\xi - \eta - k) x} = \sqrt{2 \pi \eps}
\delta_0(\xi - \eta - k)
\ee
in order to obtain the result.
\end{proof}
\noindent {\bf Remark:} In position space, the  Weyl quantisation of $\kappa(p,q) = p^m g(q)$ is given by
\be \label{pos space weyl}
(\caW_\eps \kappa) \psi (x) = (-\ii \eps)^m \sum_{j=0}^m \binom{m}{j} 2^{-j} (\partial^j g)(x) (\partial^{m-j} \psi)(x),
\ee
which can be seen directly from (\ref{weyl quant}) using integration by parts with respect to $y$,
and (\ref{epsft delta}).\\[2mm]
We will now give the momentum space version of (\ref{psi_- sol}) in a fairly explicit form.
We write $\e{-s\frac{\ii}{\eps} \hat H}$ for the unitary propagator in the Fourier picture.
We will also write $\kappa_{n+1}^{0,-}$ in the form
\[
\kappa_{n+1}^{-}(p,q) = \sum_{m=0}^{\floor{(n+1)/4}} p^{n+1-4m} \kappa_{n+1}^{4m,-}(q),
\]
according to our results from Section \ref{s3}. Combining (\ref{psi_- sol}) and Lemma \ref{weyl p^m g} now immediately shows that 
to leading order in $\eps$, and for any $n \in \N$, the component $\psi_{-,n}$ of the solution to (\ref{2level})
in the $n$-th supeardiabatic representation is given by
\be \label{transitions 1}
\epsft \psi_{-,n}(k,t) = \eps^{n} \frac{- \ii}{\sqrt{2 \pi \eps}} \e{-t \frac{\ii}{\eps} \hat H_-} \int_{-\infty}^t
\Big( \e{s \frac{\ii}{\eps} \hat H_-} J_{n+1} \e{- s \frac{\ii}{\eps} \hat H_+} \epsft{\psi}_{+,0} \Big) (k) \, \dd s
\ee
Here, the operator $J_{n+1}$ is given by $J_{n+1} = \sum_{m=0}^{\floor{(n+1)/4}} J_n^{4 m}$, with
\be \label{transitions 2}
J_{n+1}^{4m} f(k) = \int_\R \dd \eta \, \epsft{\kappa_{n+1}^{4m,-}}(k - \eta)
\left(\frac{\eta + k}{2} \right)^{n+1-4m} f(\eta).
\ee

\subsection{Transitions for constant energy levels} \label{constant energy transitions}
Although general, \eqref{transitions 1} is not  very helpful when trying to calculate actual superadiabatic transition wave functions. To make further progress, we will make two simplifying assumptions. Firstly, we will treat the high momentum regime, as discussed above. Secondly, we will assume constant energy levels. This means that in \eqref{basic V} we will take  
\be \label{const eval}
 \rho(q) = \delta, \quad \text{and} \quad \theta'(q) = \frac{\ii\gamma}{q - \ii q_{\rm c}} - 
\frac{\ii\gamma}{q + \ii q_{\rm c}} + \theta_{\rm r}(q),
\ee
with $\theta_{\rm r}$ having no singular points of order greater or equal to one in the strip 
$\{ z \in \bbC: |\Im(z)| \leq q_{\rm c} \}$. Constant $\rho$ 
has the effect of trivializing at the same time the propagator in Fourier space,
\[
\e{s \hat H_{\pm}} \epsft\psi(k) = \e{- \frac{\ii s}{2 \eps}(k^2 \pm \delta)} \epsft\psi(k),
\]
and the transformation \eqref{nat} to the natural time scale, with $\tau(q) = 2 \delta q$. 
Then, \eqref{nat scale} becomes 
\[
\partial_\tau \tilde \theta(\tau(q)) = \frac{\ii\gamma}{2 \delta q - \ii \tau_{\rm c}} - 
\frac{\ii\gamma}{2 \delta q + \ii \tau_{\rm c}} + \theta_{\rm r}(2 \delta q),
\]
with $\tau_{\rm c} = 2 \delta q_{\rm c}$. 

Now \eqref{kappa formula} together with \eqref{h_n^0} yield
\[
\kappa_n^{0,-} \approx  2 \alpha(n) \delta (-\ii)^n (n-1)! \frac{1}{(2 \delta)^n}
\left( \frac{\ii \gamma}{(q - \ii q_{\rm c})^n} -
\frac{\ii \gamma}{(q + \ii q_{\rm c})^n} \right).
\]
Using the residual theorem, the scaled Fourier transformation in $q$ of $\kappa_n^{0,-}$ is given by
\[
\epsft{\kappa^{0,-}_n} (k) = \ii \delta \alpha(n) \frac{\gamma}{(2 \delta)^n} \sqrt{\frac{2\pi}{\eps}}
\left( \frac{k}{\eps} \right)^{n-1} \e{-\frac{|k|}{\eps} q_{\rm c}}.
\]
By \eqref{transitions 2}, we find
\[
J_{n+1}^0 f(k) = \frac{\ii \gamma}{4} \alpha(n+1) \sqrt{\frac{2 \pi}{\eps}} \eps^{-n} \int
\left(\frac{k^2 - \eta^2}{4 \delta}\right)^n (k+\eta) \e{-\frac{q_{\rm c}}{\eps} |k - \eta|} f(\eta) \, \dd \eta.
\]
We now insert this into \eqref{transitions 1}. For this it is useful to pick the origin of the time axis such that the transition occurs at $t=0$. 
The evolution for $\psi_{+}$ is, to leading order, just the free Born-Oppenheimer evolution in the upper energy band;
thus indeed we {\em do not} need to solve the full, coupled system to obtain the upper wave function at $t=0$ from the one given at
some negative $t$ and vice versa, which would defeat the purpose of what we are doing.
With the above conventions, and our calculations so far, \eqref{transitions 1} reads
\be
\begin{split}
\epsft{\psi_{-,n}}(k,t) = & \frac{\gamma \alpha(n+1)}{4 \eps} \e{- \frac{\ii}{\eps} t (k^2/2 - \delta)} \times \\
& \times \int_{-\infty}^t \dd s \int \dd \eta (\eta + k)
\left(\frac{k^2 - \eta^2}{4 \delta}\right)^n \e{- \frac{q_{\rm c}}{\eps} |k-\eta|}
\e{\frac{\ii s}{2 \eps}(k^2 - \eta^2 - 4 \delta)} \epsft{\psi_{+,0}}(\eta).
\end{split}
\label{psiMinusHat}
\ee
For the moment, we are interested in $\epsft{\psi_{-,n}}(k,t)$ for $t \gg 0$. The integral in (\ref{psiMinusHat}) converges
as $t \to \infty$, and in the limit can be calculated explicitly. To do this, consider the general integral
\[
\begin{split}
  & \int_{\R^2} \dd \eta \,  \dd s \, f(\eta) \e{\frac{\ii s}{2 \eps}(k^2 - \eta^2 - 4 \delta)}
  = \\ 
  & \!\! = \int_\R \dd s \Bigg(
  \underbrace{
  \int_0^\infty \dd \eta f(\eta) \e{\frac{\ii s}{2 \eps}(k^2 - \eta^2 - 4 \delta)}
  }_{=:I_+}
  +  \underbrace{
  \int_{-\infty}^0 \dd \eta f(\eta) \e{\frac{\ii s}{2 \eps}(k^2 - \eta^2 - 4 \delta)}
  }_{ =:I_-}
  \Bigg).
  \end{split}
\]
Setting $\tilde \eta = (k^2-\eta^2-4\delta)/(2\eps)$ gives $\eta = \pm (k^2-4\delta-2\tilde\eta \eps)^{1/2}$ and $\dd \eta = \mp \eps \dd \tilde \eta (k^2-4\delta-2\tilde\eta \eps)^{-1/2}$, and so
\[
  I_+ = \eps \int^{\frac{k^2-4\delta}{2\eps}}_{-\infty} \dd \tilde\eta \frac{ f\big( +\sqrt{k^2-4\delta-2\tilde\eta \eps} \big) }
  {\sqrt{k^2-4\delta-2\tilde\eta \eps}} \e{\ii s\tilde\eta}
  =\eps \int^\infty_{-\infty} \dd \tilde\eta \frac{ f\big(+ \sqrt{k^2-4\delta-2\tilde\eta \eps} \big) }
  {\sqrt{k^2-4\delta-2\tilde\eta \eps}} \e{\ii s\tilde\eta} \chi_J(\tilde\eta),
\]
where $\chi_J$ is the characteristic function on $J=(-\infty, \tfrac{k^2-4\delta}{2\eps}]$.  Similarly,
\[
  I_- =
  \eps \int^\infty_{-\infty} \dd \tilde\eta \frac{ f\big( -\sqrt{k^2-4\delta-2\tilde\eta \eps} \big) }
  {\sqrt{k^2-4\delta-2\tilde\eta \eps}} \e{\ii s\tilde\eta} \chi_J(\tilde\eta),
\]

We now use that $\int_{\R^2} \dd s \dd \tilde\eta g(\tilde\eta) \e{\ii s \tilde\eta} = \sqrt{2\pi} \int_{\R} \dd s \check{g}(s) = 2\pi g(0)$, where $\check{g}$ is the inverse Fourier transform of $g$, to give
\[
  \int_{\R^2} \dd \eta \dd s f(\eta) \e{\frac{\ii s}{2 \eps}(k^2 - \eta^2 - 4 \delta)} =
  \begin{cases}
  2\pi\eps \frac{f(+\sqrt{k^2-4\delta}) + f(-\sqrt{k^2-4\delta})}{\sqrt{k^2-4\delta}} & \text{if } k^2-4\delta>0 \\
  0  & \text{else}.
  \end{cases}
\]

Applying these calculations to  (\ref{psiMinusHat}) reveals that, for large positive $t$, we have
\begin{align}
  \epsft{\psi_{-,n}}(k,t) &= \frac{ \pi \gamma \alpha(n+1)}{2}  \e{- \frac{\ii}{\eps} t (k^2/2 - \delta)} \nonumber \\
  & \quad \times \Big[ \big(1 + \tfrac{k}{ \sqrt{k^2-4\delta} } \big) \e{- \frac{q_{\rm c}}{\eps} |k-\sqrt{k^2-4\delta}|} \epsft{\psi_{+,0}}(\sqrt{k^2-4\delta}) \label{the formula raw} \\
  & \quad \qquad + \big(-1 + \tfrac{k}{ \sqrt{k^2-4\delta} } \big) \e{- \frac{q_{\rm c}}{\eps} |k+\sqrt{k^2-4\delta}|} \epsft{\psi_{+,0}}(-\sqrt{k^2-4\delta}) \Big] \chi_{\{k^2>4\delta\}}. \nonumber
\end{align}
The two terms in the square bracket in (\ref{the formula raw}) are clearly connected to 
positive and negative incoming momenta, respectively. The second line will be negligible if either $k < 0$ or if  $\epsft{\psi_{+,0}}$ is concentrated on the negative half axis, while the third line will be negligible if $k>0$ or if $\epsft{\psi_{+,0}}$ is concentrated on the positive half axis. This shows the intuitively obvious fact that the transmitted wave packet will travel in the same direction as the incoming one. It also shows that we can replace $|k \pm \sqrt{k^2 - 4 \delta}|$ with 
$||k|-\sqrt{k^2 - 4 \delta}|$ without changing the leading order result. 
We further streamline \eqref{the formula raw} by replacing $\alpha(n+1)$ by its asymptotic value $\frac{\sin(\pi \gamma /2)}{\pi \gamma /2}$, and by introducing 
\[
v(k) = v(k,\delta)  ={\rm sgn(k)} \sqrt{k^2 - 4 \delta}.
\]
Finally, we note that asymptotically, all the superadiabatic subspaces agree, so that (\ref{the formula raw}) actually gives the asymptotic adiabatic transition. We thus conclude:\\[2mm]
After the transition, the wave function in the initially unoccupied adiabatic subspace is given by 
\begin{equation} \label{the formula}
\epsft{\psi_{-}}(k,t) =  {\rm sgn(k)} \sin \left(\frac{\pi \gamma}{2}\right)  \e{- \frac{\ii}{\eps} t (k^2/2 - \delta)}  
\e{- \frac{q_{\rm c}}{\eps} | k-v(k)|}  \chi_{\{k^2>4\delta\}}  \left(1 + \tfrac{k}{v(k)} \right)   \epsft{\psi_{+,0}}(v(k))\,.
\end{equation}
A few comments are in order.\\[2mm]
1) The occurrence of the indicator function
$\chi_{\{k^2>4\delta\}}$ can be interpreted in terms of energy conservation: Any part of the wave packet that makes the transition obtains, in addition to its kinetic energy, the potential energy difference between the electronic energy levels. Thus, there cannot be any kinetic energy $k^2/2$ smaller than $2 \delta$. Recall also that these transitions are radiation-less: instead of being radiated away from the molecule in the form of a photon, the energy is transferred into kinetic energy of the nuclei.\\[2mm] 
2) We can read off a momentum shift from (\ref{the formula raw}). We assume that $\epsft{\psi_{+,0}}$ is a semiclassical wave function, and write $|\epsft \psi_{+,0}(k)| = \e{-M(k)/\eps}$. Let us assume for convenience that the absolute minimum $k_\ast$
of $M(k)$ is on the positive real line. We find 
\be \label{exponent}
\ln |\epsft{\psi_{-,n}}(k)| \approx - \frac{1}{\eps} (q_{\rm c}(\sqrt{v^2 + 4 \delta} - v) + M(v)).
\ee
Purely by energy conservation, 
one would expect the transition wave packet to be maximal when $v = k_\ast$. However, since $v \mapsto \sqrt{v^2 + 4 \delta} - v$ is decreasing, the minimum of $\sqrt{v^2 + 4 \delta} - v - M(v)$ is shifted
to the right. One can quantify this shift when $M$ is given explicitly.\\[2mm] 
3) The last point also shows that the term $\chi_{\{k^2>4\delta\}}$ in \eqref{the formula raw} is of little consequence in practice. 
Since $k^2 > 4 \delta$ is equivalent to $v>0$, and since only a small region around its maximum matters
for the semiclassical wave function $\epsft{\psi_{-,n}}$, we can safely leave out the factor
$\chi_{\{k^2>4\delta\}}$ in \eqref{the formula raw} without changing the leading order result.\\[2mm]
4) \eqref{the formula raw} depends on $n$ only through the convergent prefactor $\alpha(n)$. In particular, we do not need to know the value of the optimal $n$ in order to obtain the correct leading order transitions. But in order to justify (\ref{the formula raw}) it is of course important to choose $n$ such that the error terms are smaller than the leading term. Moerover, the $n$-independence of the leading order term is a special feature of the constant eigenvalues, and cannot be expected to persist in more general models.\\[2mm]
5) Let us compute the transition rate from (\ref{the formula}), in the limit of large momentum and small momentum uncertainty. We choose 
\[
\epsft{\psi_{+,0}}(k) = \exp \left( - \frac{C}{2 \eps}(k - p_0)^2 \right).
\]
When $C$ is large, the minimum in \eqref{exponent} is taken very close to $v = p_0$, implying that 
$k_{\ast} \approx \sqrt{p_0^2 + 4 \delta}$. The value of the exponent at the maximum is then given by 
\[
P_{\rm{trans}}(p_0) = 
- \frac{q_{\rm c}}{\eps}(\sqrt{p_0^2 + 4 \delta} - p_0),
\]
which is the transition probability for momentum $p_0$. For large $p_0$, we have 
$\sqrt{p_0^2 + 4 \delta} - p_0 \approx 2 \delta / p_0$, so that the transition probability 
in this regime is given by 
\[
\exp\left(-\frac{2 q_{\rm c} \delta}{p_0 \eps}\right) = \exp \left(- \frac{\tau_{\rm c}}{p_0 \eps} \right).
\]
The latter is precisely the Landau-Zener transition probability for the parameters chosen in (\ref{const eval}), cf.\ e.g.\ 
\cite{BT06}, where one has to replace $\eps$ by $\eps p_0$ throughout.\\[2mm]
6) Another formula for the asymptotic shape of a non-adiabatic scattering wave function was given (and proved) by Hagedorn and Joye in \cite{HaJo05}. We do not give the formula here, as we would have to introduce too much additional notation, and refer instead to Theorem 5.1 of \cite{HaJo05}. But we comment on the differences between their result and ours. The first difference is that while 
(\ref{the formula}) is for constant $\rho$, the work of Hagedorn and Joye covers the Landau-Zener situation.
So, a direct comparison is not possible at present, but a version of \eqref{the formula} for Landau-Zener transitions is work in progress. The advantage of the result by Hagedorn and Joye it is that is rigorous, and that it covers the asymptotic region of arbitrarily small $\eps$. In contrast, by the arguments of 
Section \ref{s4}, \eqref{the formula} is probably not asymptotically exact, although it works excellently 
for all cases that we were able to test on a computer; see below. The great practical advantage of 
(\ref{the formula}) over the results of \cite{HaJo05} is that it contains no complex contour integrals, and that it does not rely on a second order approximation to the incoming wave function. Indeed, we will see below that \eqref{the formula} provides accurate results for incoming wave functions that lie outside of the scope of the theory in \cite{HaJo05}, and where the transmitted wave function is clearly not Gaussian.\\[2mm]
7) A further advantage of \eqref{the formula} is that it only uses local information: the parameters $q_{\rm c}$, $\delta$ and $\gamma$ are determined by the derivatives of the potential at the coupling point, and the incoming wave function is only needed at the crossing time. This immediately 
suggests an algorithm for computing non-adiabatic transitions, even for complicated, non-contant energy levels. As in surface hopping models, one evolves the wave packet on the initial adiabatic surface, 
until one detects a local extremum in the coupling functions $\theta'(q_0(t))$, where $q_0$ is the center of the wave packet at time $t$. One then computes the coupling functions $\kappa_{n+1}^{0,\pm}$ from the 
local shape of the potential, and uses them in (\ref{transitions 2}). The choice of the optimal $n$ is of importance, and we suggest a way to find it in Definition \ref{opt superad} below. In 
(\ref{transitions 1}), one can use the free propagator, since the transition region is small, and the 
energy surfaces are approximately parallel to first order in $q$, as their distance is minimal at a crossing point. Then in a similar way as above, one gets an explicit transition formula. 

\subsection{Numerics}
We now show that (\ref{the formula}) is in excellent agreement with highly precise numerical solutions of the full two-band Schr\"odinger equation. For solving the latter, we use standard methods, including a symmetric Strang splitting.  We denote the projection of the numerical solution onto the lower 
eigenspace by $\phi_-(q,t)$. We compare this with the result given by (\ref{the formula}), which we denote by $\psi_-$ as before. We compare our results in the Fourier representation, calculating 
the inverse $\eps$-Fourier transform of $\phi_-$ with a standard FFT. The final time $t$ is chosen so that $\|\phi_-(q,t)\|_2$ is constant under further time evolution in the exact calculation.

For our potential we choose 
\be \label{numerics pot}
   \theta(x)=\tfrac{c}{\alpha}\arctan\big(\tanh\big(\tfrac{\alpha x}{2}\big)\big), \qquad \rho=\delta=1/2
\ee
in \eqref{basic V}. 
This gives $\theta'(q)=c/2/\cosh(\alpha q)$, with singularities closest to the real line at $\pm \ii q_c= \pm \ii \tfrac{\pi}{2\alpha}$, the residue at which gives $\gamma=-\tfrac{c}{2\alpha}$.
In particular, we take $c=-\pi/3$, $\alpha=\pi/2$, giving $q_c=1$ and $\gamma=1/3$.
The choice of $\theta$ over the case $\theta(q)=\arctan(q/\delta)$ (which would make $\theta_{\rm r} = 0$ in \eqref{nat scale}) is made to increase the rate at which the potential becomes flat. This reduces the necessary computation time for the numerical solution. 
If anything, we would expect the asymptotic results for $\theta(q)=\arctan(q/\delta)$ to be better.

Our first choice for the wavefunction in the intially occupied band is a Gaussian wave packet.
At time $t=0$, it is given by
\be \label{wavefunction 1}
   \epsft{\psi_{+,0}}(p)=(2\pi\eps)^{-1/4}\exp(-\tfrac{(p-p_0)^2}{4 \eps}).
\ee
As the crossing point is at $x=0$ by our choice of potential, $\psi_{+,0}$ is sitting right at the middle of the crossing region.  
Since the eigenvalues are constant, it may be evolved backwards to $t_0<0$ exactly on the upper level:
\be \label{initial cond}
   \epsft{\psi_{+}}(p,t_0)=(2\pi\eps)^{-1/4}\exp(\tfrac{(p-p_0)^2}{4 \eps})
   \exp(-\tfrac{i}{\eps} t_0 (p^2/2+\delta)).
\ee
We use \eqref{initial cond} along with $\epsft{\psi_{-}}(p,t_0)=0$ 
as initial conditions for our numeric solution, and take $t_0$ sufficiently negative in 
order for the wave packet to be well away from the crossing region.

\begin{figure}[h!tbp]
 \begin{center}
    \begin{tabular}{cc}
     \scalebox{0.6}{
       \includegraphics{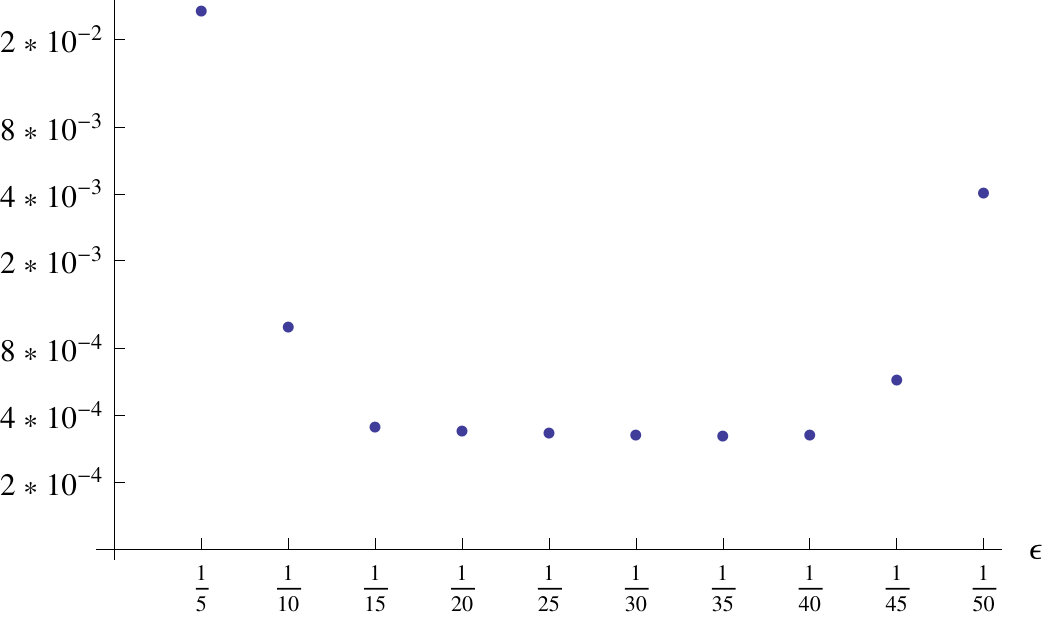}
       }
&
       \scalebox{0.6}{
       \includegraphics{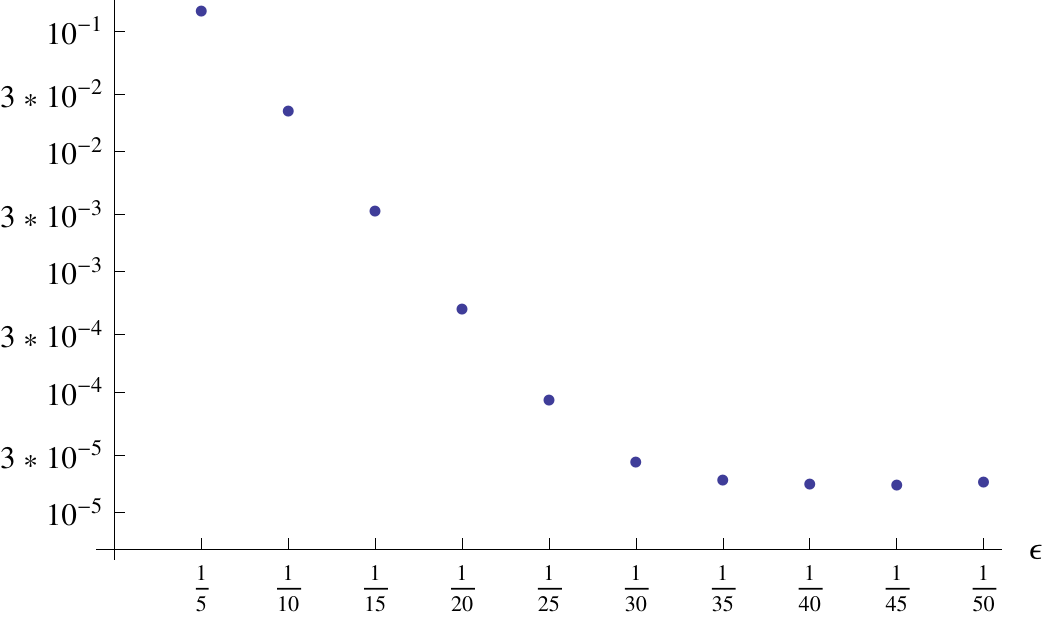}
       }
       \\
       a) $p_0=2$ & b) $p_0=5$
       \end{tabular}
\end{center}
\caption{ Relative errors between numerical results and \eqref{the formula}, for a Gaussian wave packet, with different values of $\eps$ and $p_0$; on a logarithmic scale. }
\label{fig:relativeErrors}
\end{figure}

Our numerical studies show that (\ref{the formula}) is in excellent agreement with the numerical solution for a wide range of $\eps$, ranging from as large as $1/10$ to $1/50$, at which point a further 
reduction in $\eps$ makes the numerically exact calculations very time-consuming. In Figure 
\ref{fig:relativeErrors}, we show the relative error in the $L^2$ norms between the numerical calculation and (\ref{the formula}), i.e.\  $\|\psi_- -\phi_-\|_2/\|\phi_-\|_2$. In each case, the step size in the numerical calculation was reduced until the difference between two subsequent numerical solutions was 
at least one order of magnitude smaller than the error to the solution obtained from 
(\ref{the formula}). 

Figure \ref{fig:relativeErrors} a) also shows that, as predicted in Section \ref{s4}, 
(\ref{the formula}) is not asymptotically correct for fixed $p$: After an initial increase 
in accuracy due to the decrease in $\eps$, the relative error becomes larger again as $\eps$
decreases further. That this does not affect the practical usefulness of (\ref{the formula}) 
becomes clear when we consider orders of magnitude: for $p_0=2$ and $\eps = 1/50$, we have 
$\| \psi_- \|_2 \approx 6 * 10^{-10}$, which is much smaller than is useful in practice, while 
the relative error is still excellent at about $4 * 10^{-3}$, albeit deteriorating. 
On the other hand, for $p_0=2$ and 
$\eps=1/10$, we have $\| \psi_- \|_2 \approx 0.014$, which is certainly of a physically measurable
size, and the relative error is still of the order $10^{-3}$. Finally, for $\eps = 1/5$ and 
$p_0 = 2$, we have $\| \psi_- \|_2 \approx 0.11$, and the relative error is still below $0.03$.
We see that, as is often the case in asymptotic formulae, the actual error is much better than
could be expected from the a priori error estimates: This is particularly obvious in Figure 
\ref{fig:relativeErrors} b), where the relative error initially decreases 
like $\e{-c/\eps}$ before saturating, while theory only predicts a $\sqrt{\eps}$ decrease. 
Orders of magnitude in this case range from $\| \psi_- \|_2 \approx 0.138$ for $\eps = 1/10$, with relative error below $0.025$, to $\| \psi_- \|_2 \approx 6 * 10^{-5}$ for $\eps = 1/50$, with relative error below $2 * 10^{-5}$.

\begin{wrapfigure}{r}{90mm}
 \begin{center}
 \scalebox{0.5}{
       \includegraphics{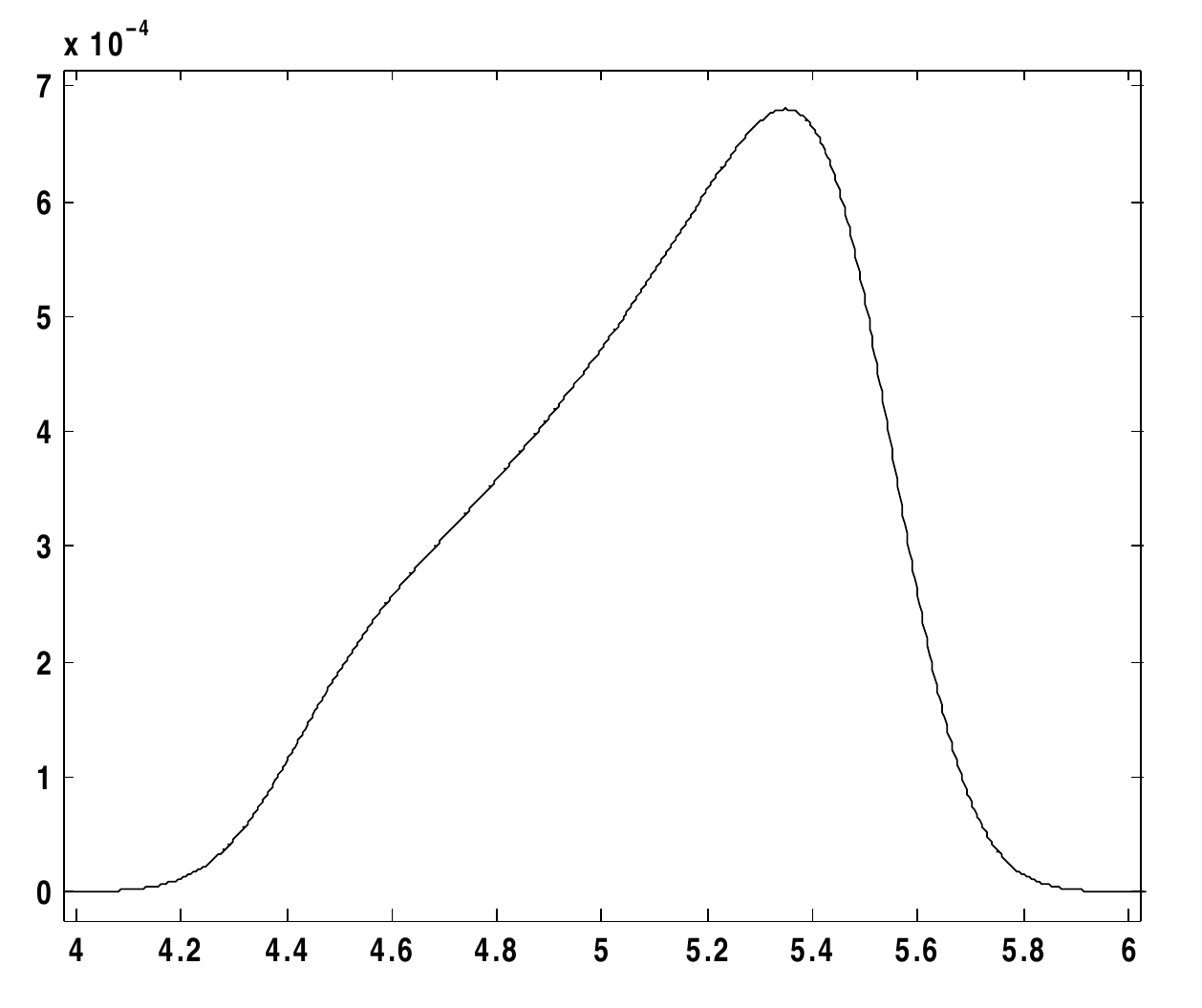}
       }
 \end{center}
\caption{Transmitted wave function for a non-Gaussian incoming wave packet.}
\label{fig:non-gauss}
\end{wrapfigure}
Our second numerical example is an incoming wave packet that is strongly non-Gaussian. We choose 
$
\epsft{\psi_{+}}(p) = \frac{1}{Z} \e{-(p-p_0)^6  / (4\eps)},
$
where $Z$ normalises the $L^2$ norm of $\epsft{\psi_{+}}$. The potential $V$ is precisely as in the 
other example, and we choose $p_0=5$ and $\eps = 1/50$. 
Figure \ref{fig:non-gauss} shows the absolute value of the transmitted wave packet in the scattering region, in Fourier representation. The relative error in the $L^2$-norm in this case is around $5 * 10^{-5}$, similar to the one in the Gaussian case. In particular, the pointwise relative error
is extremely small in the region where $\epsft{\psi_{-}}$ is concentrated; thus, while Figure 
\ref{fig:non-gauss} shows the result of applying \eqref{the formula}, the plot showing the  
numerical calculation would be indistinguishable from the one given. 
Note also that the momentum shift relative to the energy conservation value 
$\sqrt{p_0^2 + \delta^2} \approx 5.02$ is clearly visible. Obviously, the transmitted wave packet is 
strongly non-Gaussian, with a surplus of high momentum components having made the transition.

\subsection{Optimal superadiabatic transition histories}

From the time-adiabatic theory we know that in the optimal superadiabatic representation, 
the transmitted part of the wave function builds up monotonically,
and has the shape of an error function as a function of time. We are now going to show that a similar 
property holds for the Born-Oppenheimer transitions, after suitable modifications. 

Our starting point is Theorem \eqref{transitions 1}, which gives a formula for transitions in the $n$-th superadiabatic representation. Our first task is to find a precise meaning to optimality of a superadiabatic representation.
One natural idea would be to require that in the optimal superadiabatic representation the map
$t \mapsto \int |\psi_-(x,t)|^2 \dd x$ increases monotonously during the transition. But given the
complexity of formulas \eqref{transitions 1} or \eqref{psiMinusHat}, this condition
is rather difficult to check. We will instead choose another condition which, as we will argue, should be equivalent.
The basic idea is that through (\ref{transitions 1}), $\epsft \psi_-(k,t)$ is given by an integral where the
integrand is expected to be both highly oscillatory and sharply peaked in $s$ and $k$.
Thus the main contribution occurs  at points where the integrand has either a stationary phase, or a  maximal absolute value. The locations of both of these depend on $n$, and if we want to have any chance of seeing a 'nice' transition history, we have to chose (if possible) $n$ such that these locations coincide. Thus we define

\begin{definition} \label{opt superad}
Let $u = \e{s \frac{\ii}{\eps} \hat H_-} J_{n+1} \e{- s \frac{\ii}{\eps} \hat H_+} \epsft{\psi}_{+,0}$, as given in
(\ref{transitions 1}). We write $u(s,k,n) = \exp(\frac{1}{\eps}(- X(s,k,n) + \ii Y(s,k,n)))$.
Let $(s_\ast(n), k_\ast(n))$ be the location of the minimum of $X$ on the real line for given $n$.
We say that $n$ is an optimal superadiabatic representation if
$\partial_s Y (s_\ast(n), k_\ast(n)) \approx 0$.
\end{definition}

The $\approx$ sign in the above definition means that in principle, $n$ has to be integer, and thus equality might not
hold anywhere. On the one hand, this problem should be less severe when $\eps$ is extremely small, and on the other hand,
usually $J_{n+1}$ is given by an explicit formula, in which case we can interpolate and take $n$ to be real. In that case,
we will usually be able to fulfil $\partial_s Y (s_\ast(n), k_\ast(n)) = 0$ exactly.

That the above definition indeed gives 'optimal' transition histories can be seen by the following simple calculation.
Let us write
$X_{kk} = \partial_k^2 X(k_\ast,s_\ast)$ etc., and $\tilde k = k-k_\ast$, $\tilde s = s-s_\ast$. 
Expanding the exponent of \eqref{transitions 1} around $k_\ast$ and $s_\ast$ then gives
\[
\begin{split}
\e{\frac{\ii}{\eps} t H_-} \psi_{-,n}^\eps(k,t) \approx &  \frac{-\ii \eps^n}{\sqrt{2 \pi \eps}} \e{\frac{1}{\eps}
(- X(k_\ast,s_\ast) + \ii Y(k_\ast,s_\ast))  } \e{\frac{1}{2 \eps} (-X_{kk} + \ii Y_{kk}) \tilde k^2} \times
\\
& \times \int_0^t \e{\frac{1}{2 \eps} ( (-X_{ss} + \ii Y_{ss}) \tilde s^2 + 2 (-X_{ks} + \ii Y_{ks}) \tilde k \tilde s)} \, \dd s.
\end{split}
\]
Thus for $\tilde k = 0$, i.e.\ at the maximum of the transmitted wave function, we have
\[
\e{\frac{\ii}{\eps} t H_-} \psi_{-,n}^\eps(k_\ast,t) \approx  \frac{-\ii \eps^n \e{\frac{1}{\eps}
(- X(k_\ast,s_\ast) + \ii Y(k_\ast,s_\ast))  }}{2 \sqrt{X_{ss} - \ii Y_{ss}}}  \left(1 + \mathrm{erf}\left( (t - s_\ast) \sqrt{\frac{X_{ss} - \ii Y_{ss}}{2 \eps}} \right) \right).
\]
We see that the transmitted wave function, when adjusted for the propagation in the lower band, has the shape of
an error function at its maximum $k_\ast$. The only unusual feature is that this error function is actually evaluated
along a 'diagonal' in the complex plane, rather than on the real line. Nevertheless, the resulting shape will be close to monotone
unless $Y_{ss}$ is much larger than $X_{ss}$, in which case there are some oscillations around $t=0$. We do not know whether this case is likely to happen in practice. For $\tilde k \neq 0$, the error function behaviour deteriorates, but since
$ \psi_{-,n}^\eps$ is a semiclassical wave function, we are only interested in $\tilde k \sim \sqrt{\eps}$, in which
case the behaviour is still similar to the one at the maximum.

We now want to show that optimal superadiabatic representations exist in particular cases. As an example, we
pick again the situation of constant eigenvalues. We rewrite the integral in \eqref{psiMinusHat} as
\be \label{I}
I = \int_{-\infty}^t \dd s \int \dd \eta \exp \left( \frac{1}{\eps} (-M(k,\eta) + \ii \phi(\eta) + \frac{\ii}{2} s (k^2 - \eta^2 - 4 \delta)) \right).
\ee
Here, $\e{-M/\eps}$ is the combined modulus of the transition kernel originating from $J_{n+1}^0$ and the wave function, and
$\phi$ is the phase of the wave function, hence depending only on $\eta$.
We can treat more general forms of the coupling function $\kappa_n^{-}(p,q)$ than the one leading to \eqref{psiMinusHat},
including the full coupling function with all powers of $p$ included. The only requirement is that, to leading order,
$\kappa$ should be symmetric, so that its Fourier transform is real. This is generically true:  high derivatives of the pair of first order poles in the complex plane determines the shape of 
$\kappa$, and these are either symmetric or antisymmetric, giving either purely real or 
purely imaginary Fourier transforms. 

Given \eqref{I}, we now let $\eta_\ast = \eta_\ast(n)$ and
$k_\ast = k_\ast(n)$ be the place where $M(k,\eta)$ is minimal. We expand $M$ to second order around $(k_\ast,\eta_\ast)$.
In order for the phase to be also stationary, we need $\partial_\eta \phi(\eta_\ast) = s \eta_\ast$, which is the equation
determining the transition time $s_\ast$. We should keep in mind that any further explicit calculations will only make sense
for $s$ close to $s_\ast$. Now, we expand the phase to second order around $\eta_{\ast}$ as well, and compute the resulting
Gaussian integral in $\eta$. The result is
\be \label{time development}
\begin{split}
I \approx & \e{-\tfrac{1}{\eps}(M(\eta_\ast,k_\ast) - \ii \phi(\eta_\ast) + \frac{1}{2} M_{kk} \tilde k^2)} \times \\
& \times \int_{-\infty}^t \sqrt{\frac{2 \pi \eps}{M_{\eta\eta} - \ii (\phi_{\eta\eta} - s)}}
\e{\frac{1}{2\eps}\left( \ii s (\tilde k^2 + 2 k_\ast \tilde k + k_\ast^2 - \eta_\ast^2 - 4 \delta) + \frac{-(\eta_\ast (s-s_\ast) - \ii M_{k \eta} \tilde k)^2}{M_{\eta \eta}
- \ii (\phi_{\eta \eta} - s)}\right)} \, \dd s.
\end{split}
\ee
Above, we use again the notation $\tilde k = k - k_\ast$, $M_{kk} = \partial_k^2 M(k_\ast,\eta_\ast)$ etc. As before,
we concentrate on the case $\tilde k = 0$. Let us assume that we can find $n$ such that
$k_\ast^2 - \eta_\ast^2 - 4 \delta = 0$. It is then easy to check that for the remaining integrand, both real and imaginary
part of the exponent are stationary at $s = s_\ast$, and that $s_\ast$ is a maximum of the real part. Thus for constant eigenvalues,
we can find an optimal superadiabatic representation if we can solve the equations
\be \label{equations}
\partial_k M(k,\eta) = 0, \quad \partial_\eta M(k,\eta) = 0, \quad k^2 - \eta^2 = 4 \delta
\ee
simultaneously, and if the resulting pair $(k,\eta)$ is a minimum of $M$. Note that as the above equations also depend on $n$,
we are solving a system of three equations with three free variables, which means that we can hope for a solution. Note also that the third equation is connected to energy conservation: $\eta$ is the incoming, $k$ the outgoing momentum and $\delta$ is the gap between the energy levels.

We specialise further in order to solve \eqref{equations}. We choose the upper band wave function to be a Gaussian wave packet with momentum $p_0$, so that
\be \label{upper initial sigma}
\epsft{\psi_{+,0}}(\eta) = \frac{1}{(\sigma^2 \pi  \eps)^{1/4}}\exp \left(- \frac{1}{2 \sigma^2 \eps} (\eta - p_0)^2 \right).
\ee
Choosing $\epsft{\psi_{+,0}}(\eta)$ real-valued amounts to choosing the packet to be at $x=0$ at time $t=0$;
putting the avoided crossing at $x=0$ in addition ensures that the transition occurs at time $s_\ast = 0$.
Now, using the integrand in \eqref{psiMinusHat}, we get, to leading order,
\[
M(k,\eta) = - n \eps ( \ln (k^2 - \eta^2) - \ln (4 \delta)) + q_{\rm c} |k - \eta|  + \frac{(\eta - p_0)^2}{2 \sigma^2}.
\]
By the third equation in (\ref{equations}), we know $k>\eta$, which
removes the absolute value above. Taking derivatives,
\[
0 = 2 n \eps \frac{\eta}{k^2 - \eta^2} - q_{\rm c} + \frac{\eta - p_0}{\sigma^2}, \quad  0 = - 2 n \eps \frac{k}{k^2 - \eta^2} + q_{\rm c},
\]
which together with the equation $k^2 - \eta^2 = 4 \delta$ lead to
\be \label{nonlin eq}
k = \sqrt{\eta^2 + 4 \delta}, \qquad
\eta = k \left(1 - \frac{\eta - p_0}{\sigma^2 q_{\rm c}} \right), \qquad
n  = \frac{2 \delta q_{\rm c}}{\eps k}.
\ee
The first two equations are independent of $n$ and $\eps$, and determine $\eta_\ast, k_\ast$. This is a special feature of the constant eigenvalue situation, and should not be expected in general. The third equation determines the
optimal superadiabatic representation. Note that $n$ is connected to the optimal superadiabatic $n$ for
the time-adiabatic situation: there, we have $n_{\mathrm{ta}} = 2 \delta q_{\rm c} / \eps$, with the 'momentum'
(i.e. speed on the time axis) normalized to one. But the tricky point that can't be easily guessed is which value of $k$ to pick in the formula for $n_\ast$; the naive guess of using the incoming momentum $p_0$ would be totally wrong.
The more sophisticated guess of using the mean momentum at the
crossing point, $(\eta_\ast + k_\ast)/2$ would be closer for small $\delta$, since the true value fulfils
$k_{\ast} = \frac{1}{2} (k_\ast + \sqrt{\eta^2 + 4 \delta})$; but it would still be far off for finite $\delta$,
which are of main interest here.

\begin{figure}[h!tbp]
 \begin{center}
    \begin{tabular}{ccc}
     \scalebox{0.25}{
       \includegraphics{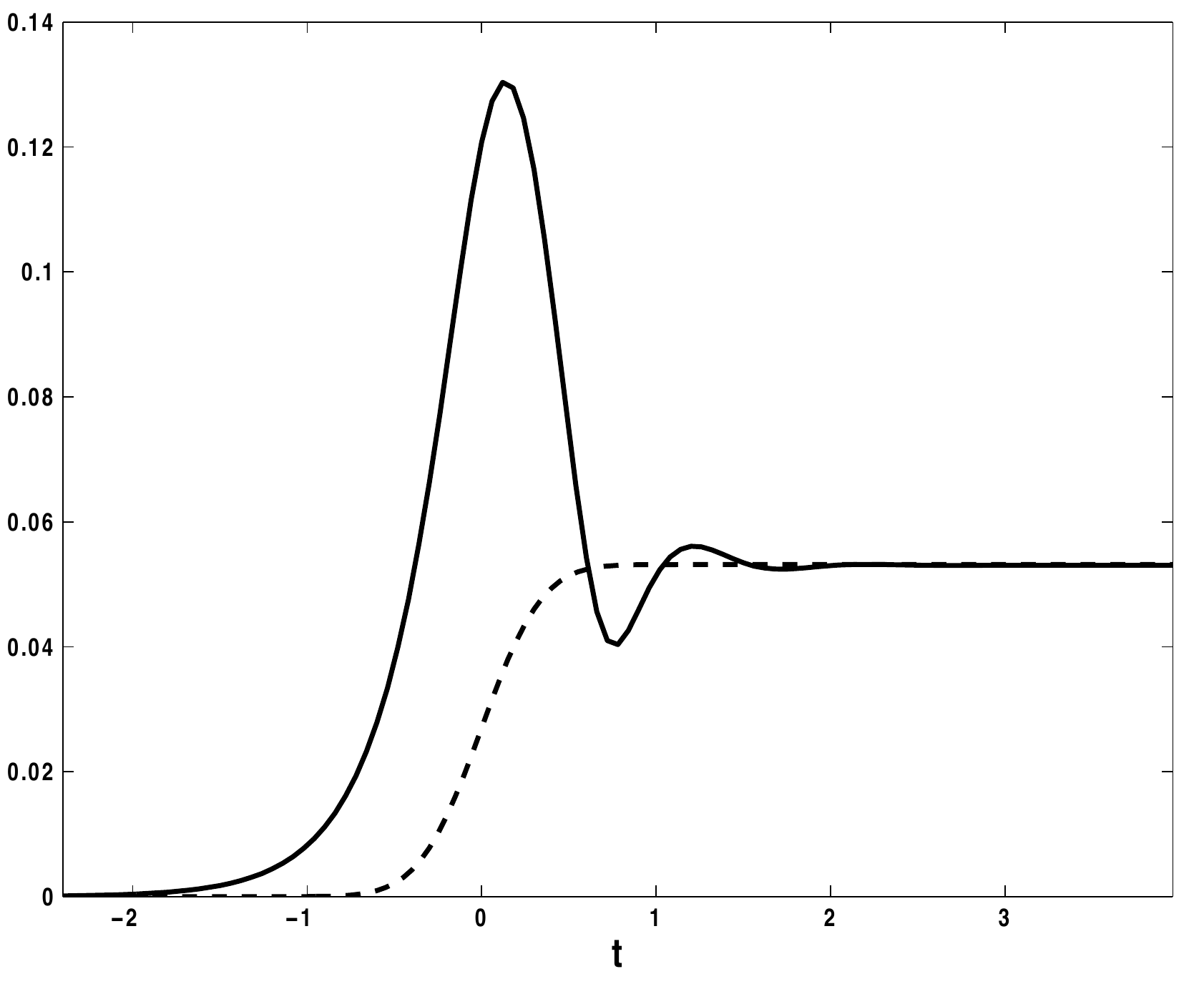}
       }
&
       \scalebox{0.25}{
       \includegraphics{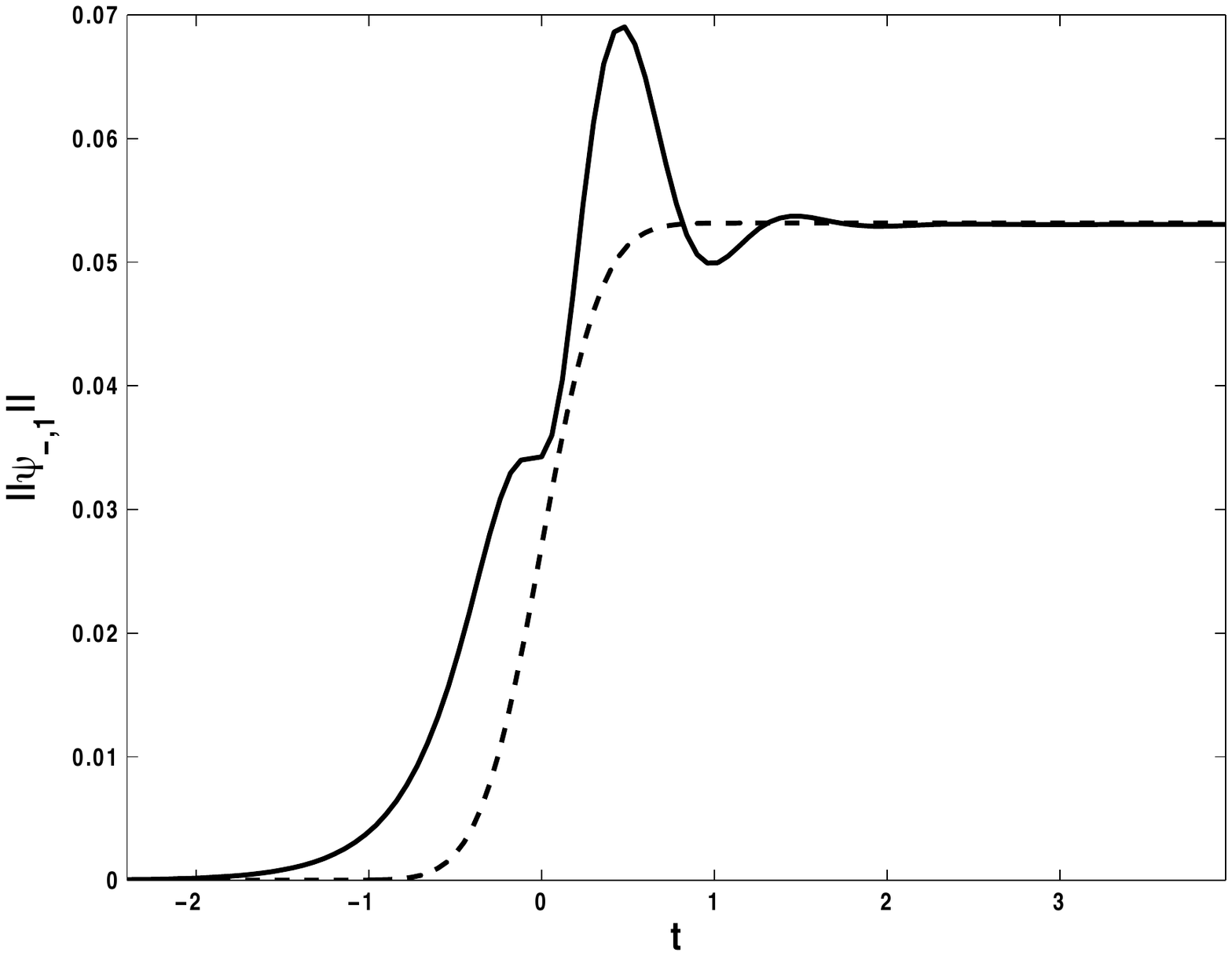}
       }
&
       \scalebox{0.25}{
       \includegraphics{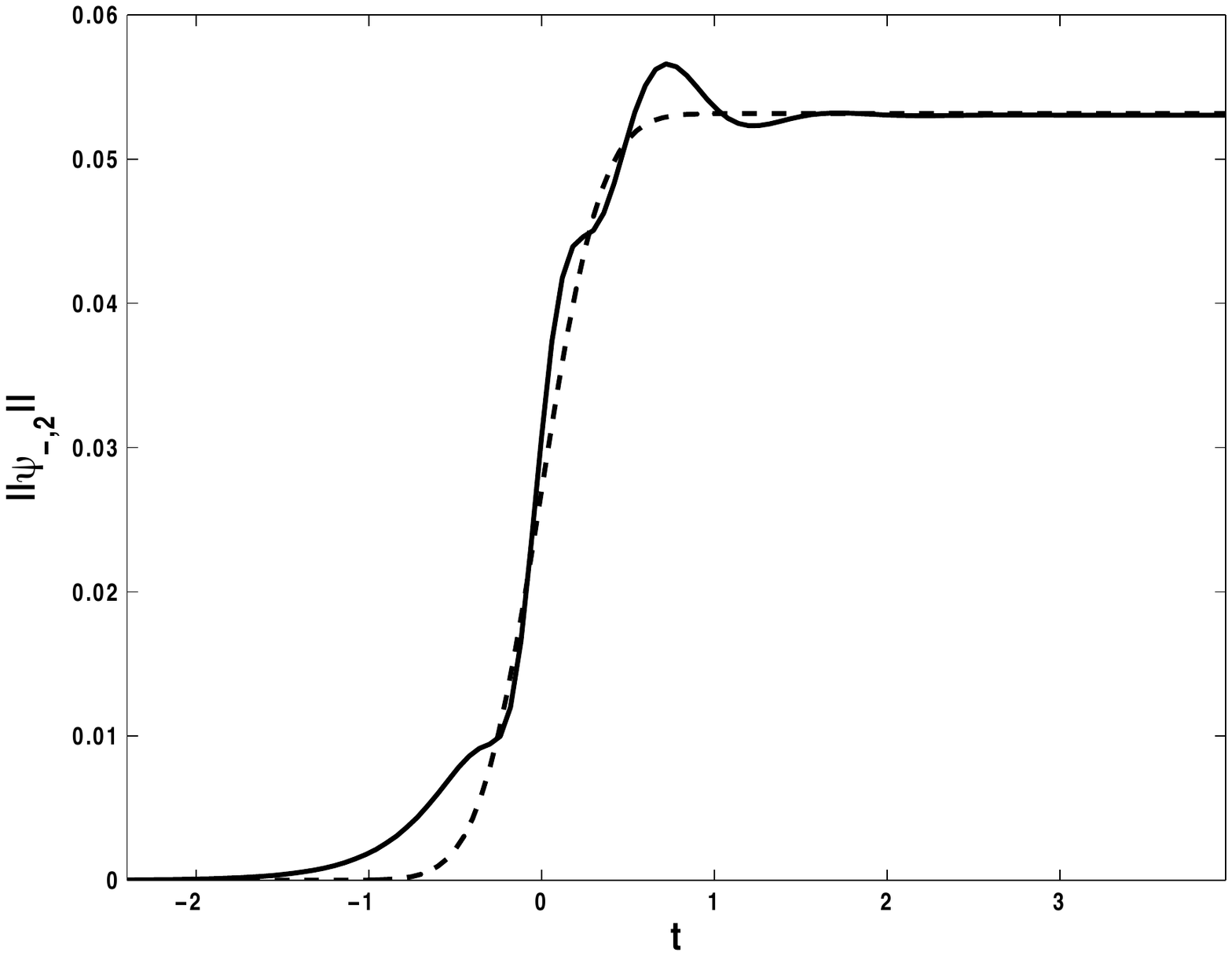}
       }
       \\[-2mm]
       \scriptsize $\|\psi_{-,0}(.,t)\|_2$ & \scriptsize $\|\psi_{-,1}(.,t)\|_2$ & \scriptsize $\|\psi_{-,2}(.,t)\|_2$
       \\[2mm]
       
       \scalebox{0.25}{
       \includegraphics{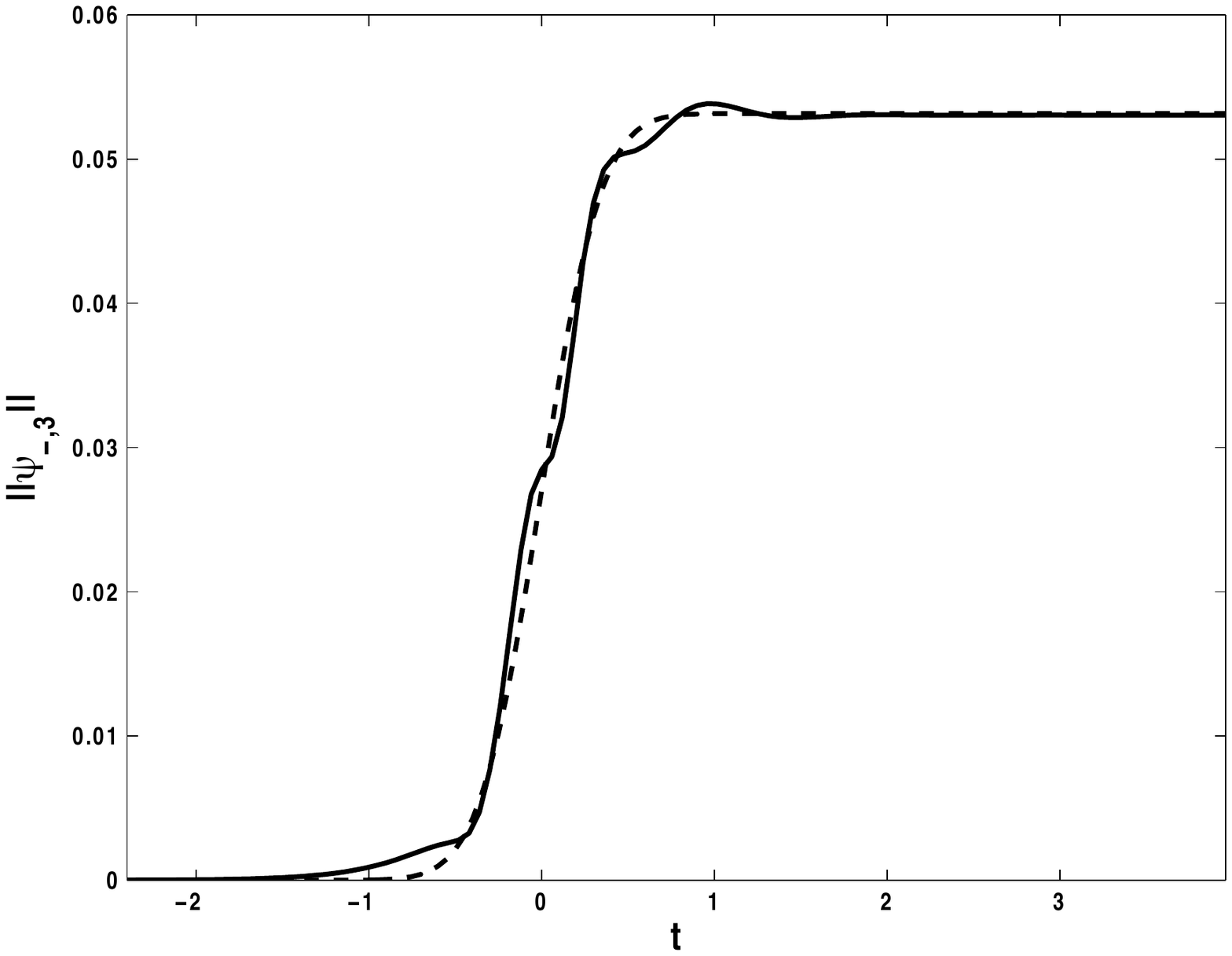}
       }
&
       \scalebox{0.25}{
       \includegraphics{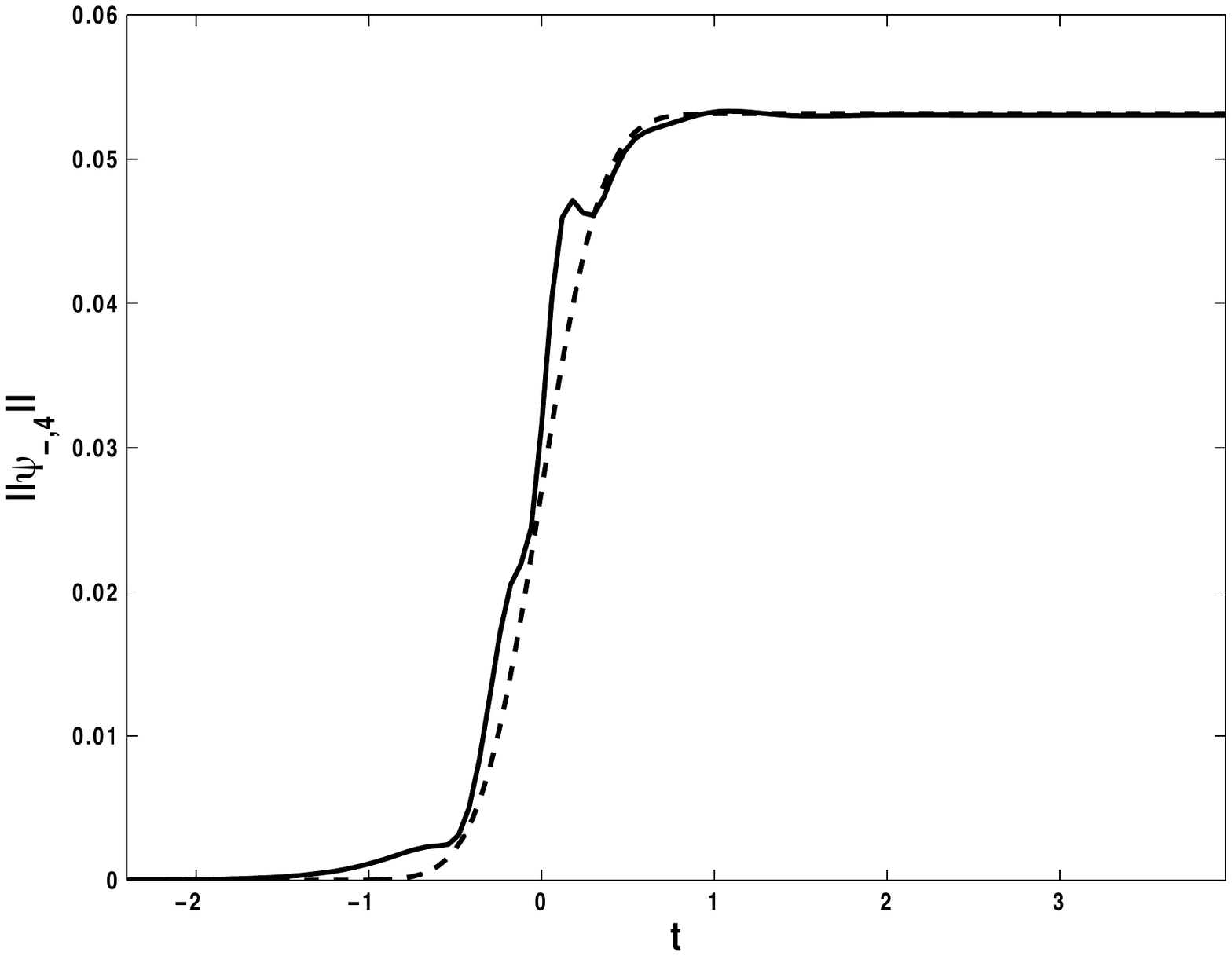}
       }
&
       \scalebox{0.25}{
       \includegraphics{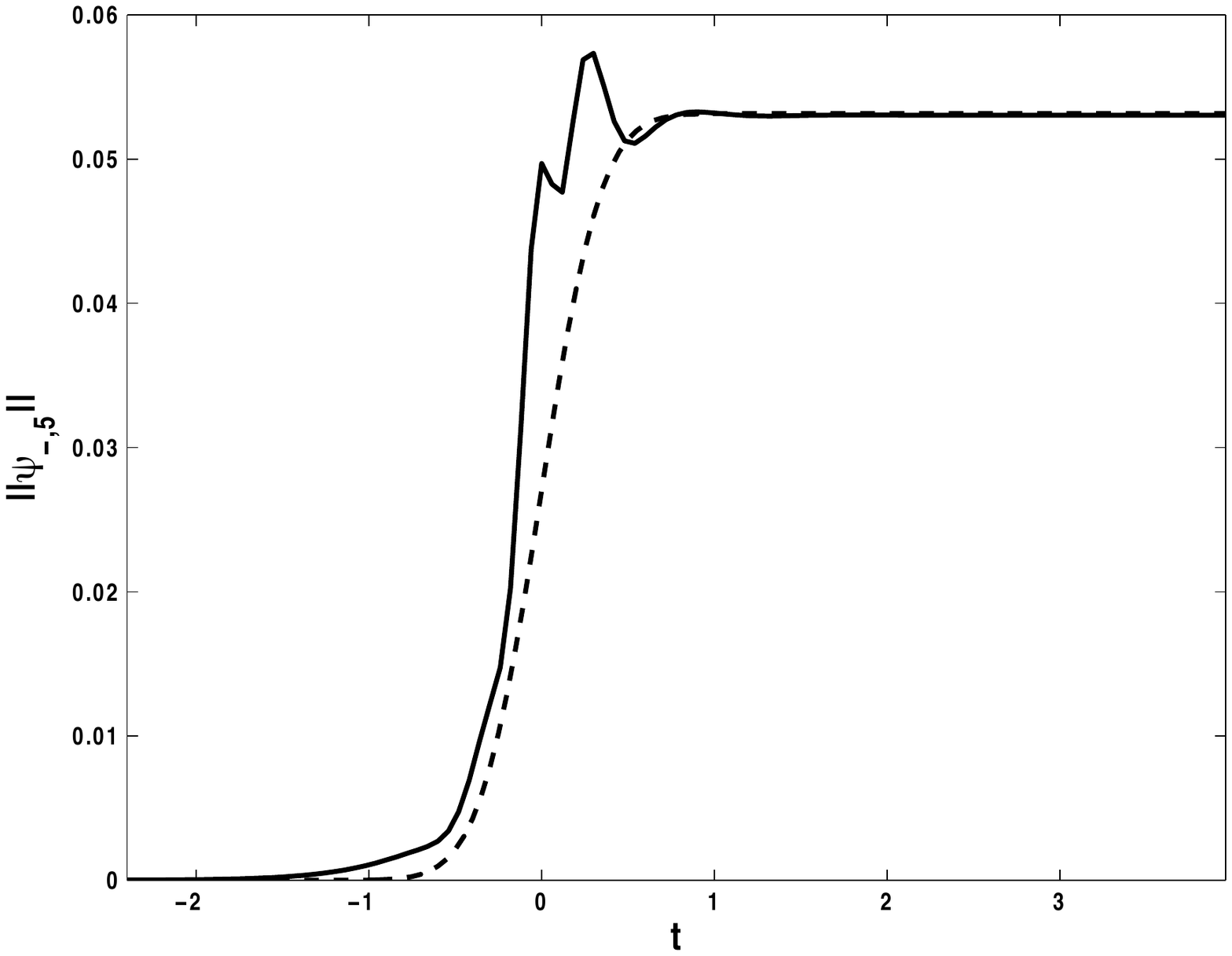}
       }
     \\[-2mm]
       \scriptsize $\|\psi_{-,3}(.,t)\|_2$ & \scriptsize $\|\psi_{-,4}(.,t)\|_2$ & 
       \scriptsize $\|\psi_{-,5}(.,t)\|_2$
       \end{tabular}
\end{center}
\caption{ Time development of the norm of the transmitted wave function, from the adiabatic up to the fifth superadiabatic representation. The dashed line is the theoretical prediction of the optimal superadiabatic transition history. }
\label{fig:histories}
\end{figure}

We close by comparing the effective formula \eqref{time development} 
to numerical results for the transition histories in various superadiabtic 
representations. In order to obtain the optimal superadiabatic representation at 
both low and close to integer values of $n$, we have to choose the parameters somewhat carefully.  
Our choices are 
the potential (\ref{numerics pot}) with $c = -\pi/3$, $\alpha = 2 \pi/5$ and $\delta = 3/32$. 
This gives $q_{\rm c} = 5/4$ and $\gamma = 5/12$. 
We take $\eps = 0.02923$, and in the incoming wave function \eqref{upper initial sigma},
we take $\sigma^2 = 2$ and $p_0 = 2.5$. Solving \eqref{nonlin eq} then yields 
$\eta_\ast \approx  2.57$, $k_\ast \approx 2.64$, and $n_\ast \approx 3.04$. Thus the optimal
superadiabatic representation should be the third, and this is clearly confirmed by numerical 
simulations. 
Figure \ref{fig:histories} shows the $L^2$ norm 
of the transmitted wave function (calculated by a numerical solution of the Schr\"odinger equation), 
in the adiabatic and all superadiabatic representations up 
to $n=5$, as a function of time $t$. The dashed line is the prediction of formula 
\eqref{time development} with $\eta_\ast$ and $k_\ast$ from above inserted. We see that 
indeed the optimal superadiabatic representation is at $n=3$, while below $n=1$ and above $n=4$
oscillations grow. The reader should note the similarity with the plots shown in \cite{LiBe91}. This is
another confirmation that the time-adiabatic approximation is appropriate to qualitatively understand
the mechanism of non-adiabatic transitions, as far as the population of the lower level is 
concerned. On the other hand, in order to obtain quantitatively correct results, or more detailed information about the transmitted wave packet (e.g.\ momentum spread of phase shift), the full quantum mechanical treatment, as given in the present work, is indispensable. 

\medskip
{\footnotesize
{\bf Acknowledgment:} V.B.\ is supported by the EPSRC fellowship EP/D07181X/1.
}


\begin{thebibliography}{10}

\bibitem{Be90}
M.~Berry.
\newblock Histories of adiabatic quantum transitions.
\newblock {\em P Roy Soc Lond A Mat}, 429(1876):61--72, Jan 1990.

\bibitem{BeLi93}
M.~Berry and R.~Lim.
\newblock Universal transition prefactors derived by superadiabatic
 renormalization.
\newblock {\em J Phys A-Math Gen}, 26(18):4737--4747, Jan 1993.

\bibitem{BG08}
V.~Betz and B.~Goddard. 
\newblock Superadiabatic representations for Born-Openheimer 
transitions through an avoided Landau-Zener crossing. 
\newblock in preparation.

\bibitem{BT05-1}
V.~Betz and S.~Teufel.
\newblock Precise coupling terms in adiabatic quantum evolution.
\newblock {\em Ann. Henri Poincar{\'e}}, 6(2):217--246, 2005.

\bibitem{BT05-2}
V.~Betz and S.~Teufel.
\newblock Precise coupling terms in adiabatic quantum evolution: the generic
 case.
\newblock {\em Comm. Math. Phys.}, 260(2):481--509, 2005.

\bibitem{BT06}
V.~Betz and S.~Teufel. 
\newblock Landau-Zener formulae from adiabatic transition histories. 
\newblock Lecture Notes in Physics {\bf 690}, p.\ 19-32,  Springer, 2006.

\bibitem{BO27}
M.~Born and R.~Oppenheimer.
\newblock Zur {Q}uantentheorie der {M}olekeln.
\newblock {\em Ann.\ Phys.\ (Leipzig)}, 84:457--484, 1927.

\bibitem{BH06}
I.~Burghardt and J.~Hynes.
\newblock Excited-state charge transfer at a conical intersection: Effects of
 an environment.
\newblock {\em Journal of Physical Chemistry A}, 110(40):11411--11423, 2006.

\bibitem{CGB05}
L.~S. Cederbaum, E.~Gindensperger, and I.~Burghardt.
\newblock Short-time dynamics through conical intersections in macrosystems.
\newblock {\em Physical Review Letters}, 94(11):113003, 2005.

\bibitem{DOS78}
Y.~N. Demkov, V.~N. Ostrovskii, and E.~A. Solov'ev.
\newblock Two-state approximation in the adiabatic and sudden-perturbation
 limits.
\newblock {\em Phys. Rev. A}, 18(5):2089--2096, Nov 1978.

\bibitem{DiSj99} M.\ Dimassi and J.\ Sj\"ostrand, \newblock 
Spectral Asymptotics in the Semi-Classical
Limit, \newblock Cambridge University Press, 1999.

\bibitem{Hag80}
G.~A. Hagedorn.
\newblock A time dependent {B}orn-{O}ppenheimer approximation.
\newblock {\em Comm. Math. Phys.}, 77(1):1--19, 1980.

\bibitem{Ha94}
G.~A. Hagedorn.
\newblock Molecular propagation through electron energy level crossings.
\newblock {\em Mem. Amer. Math. Soc.}, 111(536):vi+130, 1994.

\bibitem{HJ01}
G.~A. Hagedorn and A.~Joye.
\newblock A time-dependent {B}orn-{O}ppenheimer approximation with
 exponentially small error estimates.
\newblock {\em Comm. Math. Phys.}, 223(3):583--626, 2001.

\bibitem{HaJo04}
G.~A. Hagedorn and A.~Joye.
\newblock Time development of exponentially small non-adiabatic transitions.
\newblock {\em Comm. Math. Phys.}, 250(2):393--413, 2004.

\bibitem{HaJo05}
G.~A. Hagedorn and A.~Joye.
\newblock Determination of non-adiabatic scattering wave functions in a
 born-oppenheimer model.
\newblock {\em Ann. Henri Poincar{\'e}}, 6(5):937--990, 2005.

\bibitem{HJ06-2}
G.~A. Hagedorn and A.~Joye.
\newblock Recent results on non-adiabatic transitions in quantum mechanics.
\newblock In N.~Chernov, Y.~Karpeshina, I.~Knowles, R.~Lewis, and R.~Weikard,
 editors, {\em Recent Advances in Differential Equations and Mathematical
 Physics.}, volume 412 of {\em AMS Contemporary Mathematics Series}, pages
 183--198. to appear, 2006.

\bibitem{HJ06}
G.~A. Hagedorn and A.~Joye.
\newblock Mathematical analysis of {B}orn-{O}ppenheimer approximations.
\newblock In {\em Proceedings of the 'Spectral Theory and Mathematical Physics'
 Conference in Honor of Barry Simon, AMS Proc.\ of Symposia in Pure Math.} to
 appear, 2007.

\bibitem{JKP91}
A.~Joye, H.~Kunz, and C.~E. Pfister.
\newblock Exponential decay and geometric aspect of transition probabilities in
 the adiabatic limit.
\newblock {\em Annals of Physics}, 208(2):299--332, 1991.

\bibitem{Lan65}
L.~D. Landau.
\newblock {\em Collected Papers of L.D. Landau}.
\newblock Pergamon Press, Oxford, 1965.

\bibitem{LT05}
C.~Lasser and S.~Teufel.
\newblock Propagation through conical crossings: an asymptotic semigroup.
\newblock {\em Comm. Pure Appl. Math.}, 58(9):1188--1230, 2005.

\bibitem{LiBe91}
R.~Lim and M.~Berry.
\newblock Superadiabatic tracking of quantum evolution.
\newblock {\em J Phys A-Math Gen}, 24(14):3255--3264, Jan 1991.

\bibitem{Lon28}
F.~London.
\newblock {\"U}ber den {M}echanismus der {H}om{\"o}opolaren {B}indung.
\newblock In P.~Debye, editor, {\em {P}robleme der {M}odernen {P}hysik}.
 Herzel, Leipzig, 1928.

\bibitem{Ma02}
A.~Martinez. 
\newblock An Introduction to Semiclassical and Microlocal Analysis.
\newblock Springer, 2002.

\bibitem{MS02}
A.~Martinez and V.~Sordoni.
\newblock A general reduction scheme for the time-dependent
 {B}orn-{O}ppenheimer approximation.
\newblock {\em C. R. Math. Acad. Sci. Paris}, 334(3):185--188, 2002.

\bibitem{RRZ89}
T.~S. Rose, M.~J. Rosker, and A.~H. Zewail.
\newblock Femtosecond real-time probing of reactions. iv. the reactions of
 alkali halides.
\newblock {\em The Journal of Chemical Physics}, 91(12):7415--7436, 1989.

\bibitem{So03}
V.~Sordoni.
\newblock Reduction scheme for semiclassical operator-valued schr{\"o}dinger
 type equation and application to {\ldots}.
\newblock {\em Communications in Partial Differential Equations}, Jan 2003.

\bibitem{ST01}
H.~Spohn and S.~Teufel.
\newblock Adiabatic decoupling and time-dependent {B}orn-{O}ppenheimer theory.
\newblock {\em Comm. Math. Phys.}, 224(1):113--132, 2001.
\newblock Dedicated to Joel L. Lebowitz.

\bibitem{Teu03}
S.~Teufel.
\newblock {\em Adiabatic perturbation theory in quantum dynamics}, volume 1821
 of {\em Lecture Notes in Mathematics}.
\newblock Springer-Verlag, Berlin, 2003.

\bibitem{Yar01}
D.~Yarkony.
\newblock Conical intersections: The new conventional wisdom.
\newblock {\em Journal of Physical Chemistry A}, 105(26):6277--6293, 2001.

\bibitem{Zen32}
D.~Zener.
\newblock Non-adiabatic crossings of energy levels.
\newblock {\em Proc. Roy. Soc. London}, 137:696--702, 1932.

\bibitem{Zew94}
A.~H. Zewail.
\newblock {\em Femtochemistry: Ultrafast Dynamics of the chemical bond}.
\newblock World Scientific, New York, 1994.

\end{thebibliography}
\end{document}